\documentclass[12pt]{article}

\usepackage{placeins}
\usepackage{pifont}
\usepackage{graphicx}
\usepackage{booktabs}
\usepackage{multirow}
\usepackage{amsmath}
\usepackage{amssymb}
\usepackage{lmodern}
\usepackage{rotating}
\usepackage{float}
\usepackage{lscape}
\usepackage{mathtools}
\usepackage{amsthm}
\usepackage{longtable}
\usepackage[T1]{fontenc}
\usepackage[a4paper,left=2.5cm,right=2.5cm,top=3cm,bottom=3cm]{geometry}
\usepackage[flushleft]{threeparttable}

\usepackage[style=authoryear-comp, backend=bibtex, maxcitenames=2, maxbibnames=99 ,dashed=false]{biblatex}
\bibliography{self_matching}

\usepackage{hyperref}

\newtheorem*{theorem}{Theorem}

\begin{document}

\title{The Symmetric Pair Matching Design: A Self-Controlled Method with Automatic Adjustment for Time Effects}
\date{}
\author{Robin Denz, Filippo Saatkamp, Katharina Meiszl, Nina Timmesfeld \\ \\ Ruhr-University Bochum \\ Department of Medical Informatics, Biometry and Epidemiology}

\maketitle

\begin{abstract}
	Self-controlled study designs eliminate confounding by individual-level characteristics that remain constant during the observation time and are thus widely used in pharmacoepidemiology and vaccine safety research. However, existing methods remain vulnerable to time effects, including temporal trends and seasonality in the exposure or outcome, unless these are explicitly modeled or controlled through the study design. We introduce the symmetric pair matching design (SPM), a novel self-controlled method that combines design-based adjustment for time effects with automatic control of time-invariant confounding. Unlike previous design-based approaches that account for time effects, SPM uses observation time both before and after event occurrence, thereby retaining a larger proportion of the available information. We derive the theoretical properties of the method and evaluate its finite-sample performance through simulations. In the simulations, SPM produced unbiased estimates in the presence of temporal trends in both the outcome and exposure, while maintaining greater or comparable statistical efficiency than existing approaches. SPM extends the family of self-controlled methods by providing design-based control of time effects without requiring explicit modeling thereof. By combining robustness to temporal confounding with efficient use of observation time, SPM offers a practical alternative for observational studies with transient exposures. An accompanying R package is provided to facilitate its usage in applied research.
\end{abstract}

\newpage

\section{Introduction} \label{chap::introduction}

Confounding remains one of the biggest issues for causal effect estimation with observational data. Many analytic approaches, including regression adjustment, inverse probability weighting or matching methods, rely on the assumption that all relevant confounders are known and measured without error \parencite{Hernan2020}. In practice, this assumption is frequently violated, resulting in potential biases in causal effect estimation. Self-controlled study designs partially adress this issue, by using each individual as their own control at different times, thereby eliminating confounding by all characteristics that remain constant during the observation period, measured and unmeasured \parencite{Bots2025a}.
\par\medskip
Among such designs are the self-controlled case series (SCCS) design \parencite{Farrington1995}, in which the observation time of an individual is divided based on the exposure status and the case-crossover (CCO) design \parencite{Maclure1991}, which additionally anchors its estimation on the occurrence of the first event. These designs are applicable when the exposure of interest is transient, meaning that its effect is only temporary. As such, these designs are particularly useful for pharmacoepidemiological \parencite{Hallas2014, Bots2025a} or vaccine safety research \parencite{Nie2022}.
\par\medskip
Although self-controlled methods inherently control for time-invariant confounding, they remain susceptible to confounding by both time-varying individual-level factors and to general time trends in the exposure or outcome. The CCO design, for example, has been shown to be biased if exposure trends are present \parencite{Suissa1995, Greenland1996} Furthermore, if both the exposure and the outcome exhibit seasonal trends simultaneously, both the SCCS and CCO method may produce biased estimates. These ``time effects'' have been studied extensively using both simulations and empirical analyses \parencite{Suissa1995, Greenland1996, Wang2011, Bateson1999, Takeuchi2018, Mostofsky2018, Dong2020}.
\par\medskip
Existing self-controlled methods therefore require explicit modeling or design-based control of such time effects. In the SCCS method age or season effects are usually adjusted using parametric modelling \parencite{Farrington2018} or through the use of smoothing splines \parencite{GhebremichaelWeldeselassie2017}. Although the spline based extension is more flexible than regular parametric adjustment, it still requires multiple modeling assumptions. In contrast, extentions of CCO design, such as the Case-Time Control design (CTC) \parencite{Suissa1995} and the Case-Case Time Control design (CCTC) \parencite{Wang2011}, do not require such explicit modeling of time effects. Here, the design itself is constructed so that time effects cancel out naturally. However, unlike the SCCS design, these methods only use a small fraction of the observation time and are thus generally less efficient in scenarios where the SCCS could also be used \parencite{Takeuchi2018}.
\par\medskip
In this article we introduce the symmetric pair matching (SPM) design, a new self-controlled method that automatically adjusts for time-invariant confounding, while additionally providing design-based adjustment for time effects. In contrast to the CTC \parencite{Suissa1995} and CCTC methods \parencite{Wang2011}, it more closely resembles the SCCS design \parencite{Farrington1995} than the CCO design \parencite{Maclure1991}, because it allows usage of time before and after event occurrences. Thus, a larger proportion of the observation time may be used when appropriate, providing a substantive gain in efficiency. To make the method easy to apply for users, we also offer a software implementation in the form of the \texttt{SPMD} R package, which is freely available on Github (\url{https://github.com/RobinDenz1/SPMD}).
\par\medskip
First, we briefly describe some relevant existing self-controlled methods. We then introduce the new method and present a brief simulation study comparing it to current standards. Finally, we discuss the advantages and disadvantages of SPM.

\section{Background}

\subsection{Setting and Notation}

We consider an observational cohort study consisting of $n$ individuals who are followed over time. During follow-up, individuals may experience recurrent outcome events and may also become exposed to a transient risk factor of interest. Let $A_i(t)$ denote the exposure status of individual $i$ at time $t$, where $A_i(t)$ equals 1 only at the time of exposure and is 0 otherwise. Throughout this paper, we assume that the exposure has a constant multiplicative effect on the outcome incidence during a constant risk period and no lingering effects thereafter. Although multiple exposure periods per individual are permitted, we primarily describe the methods for the simpler setting in which each individual experiences at most one exposure period, to simplify notation and presentation (see appendix~\ref{appendix::exposure_times}).
\par\medskip
The outcome is assumed to be recurrent, meaning that an individual may experience multiple events during follow-up. Complete information on exposure times and timings of outcome events is assumed to be available between time zero and some maximum follow-up time $t_{max}$, or until individual loss-to-follow-up (right-censoring). Time zero may correspond to the date of study inclusion for prospective studies or to the first time an individual was eligible, a relevant diagnosis date or similar dates of interest. This method is motivated by a study of \textcite{Meiszl2026}, which aimed to show whether administration of an influenza vaccination ($A_i(t)$) raised the probability of a disease-related emergency hospitalization in individuals with immune-mediated diseases.

\subsection{Target Estimand} \label{chap::target_estimand}

Our goal is to estimate the causal exposure effect during a known risk period after exposure. To define this quantity, we use the potential outcomes framework \parencite{Neyman1935, Rubin1974}. Let $T_A$ denote an individual's observed exposure time, assuming that each individual is exposed once at most, and let $\tau$ denote the prespecified duration of the post-exposure risk period. Let $N(t)$ denote the observed counting process of event occurences. For exposed individuals, let $N^0(t)$ denote the counterfactual counting process that would have been observed if the exposure at time $T_A$ had been prevented through an external intervention. Under the consistency assumption \parencite{Hernan2016, Pearl2018}, the observed counting process equals the potential counting process under the observed exposure history.
\par\medskip
The target estimand is the average causal risk ratio among exposed individuals, comparing the expected number of events occuring during the risk period under exposure and without exposure at the observed exposure times:

\begin{equation} \label{eq::target_estimand}
	\exp(\theta) =\frac{E\!\left(N(T_A+\tau)-N(T_A) \mid T_A < t_{max}\right)}{E\!\left(N^{0}(T_A+\tau)-N^{0}(T_A) \mid T_A < t_{max} \right)}.
\end{equation}

Similar estimands have been defined specifically for the SCCS \parencite{Etievant2026} and CCO designs \parencite{Shahn2023}. Because the observed counting process corresponds to the numerator of (\ref{eq::target_estimand}) under consistency, the primary estimation challenge is the identifaction of the denominator. Several causal identifiability assumptions need to be made to make this possible using observable data. Those are: (a) the previously mentioned \emph{counterfactual consistency} assumption, stating that the observed outcome process corresponds to the potential outcome process under the actually experienced exposure history \parencite{Hernan2016, Pearl2018}; (b) the \emph{positivity} assumption, requiring that exposure and non-exposure both occur with positive probability \parencite{Westreich2010}; (c) the \emph{no interference} assumption, stating that one individuals exposure status does not affect the potential outcome of other individuals \parencite{Naimi2015} and (d) \emph{exchangeability} \parencite{Sarvet2020}, stating that the counterfactual outcome process under no exposure is independent of the exposure process. These assumptions are described in more detail elsewhere \parencite{Hernan2020, Guo2014}. Additional assumptions specific to the discussed estimators are introduced in the following section.

\subsection{Existing Self-Controlled Methods}

Several self-controlled study designs have been proposed to estimate the causal effect of a transient exposures on acute outcomes. Although they share the same underlying principle, they differ in how they identify the counterfactual outcome process under no exposure. Below we briefly review the methods most relevant to the present work. This review is not meant to be exhaustive; comprehensive reviews are provided by \textcite{Takeuchi2018} or \textcite{Bots2025a}. 
\par\medskip
The SCCS design compares the incidence of events during prespecified post-exposure risk periods with the incidence of the remaining observation time within the same individual \parencite{Farrington1995}. As in all self-controlled designs, all fixed measured and unmeasured confounders are controlled for by design, because each individual serves as their own control. Identification requires that the occurrence of the outcome does not influence subsequent exposure, observation time or censoring. These assumptions have partially been relaxed through several extensions of the original SCCS method \parencite{Farrington2011, Petersen2016}. In its simplest form, the SCCS additionally assumes a constant outcome probability over the observation period. Extensions allow time effects through explicit parametric modelling \parencite{Farrington2018} or the use of smoothing splines \parencite{GhebremichaelWeldeselassie2017}. Although these approaches substantially increase flexibility, they require the analyst to specify the functional form or to choose tuning parameters, such as the number and location of spline knots.
\par\medskip
The CCO design compares the exposure status immediately preceding an event (the case period) with the exposure status during one or more control periods for the same individual~\parencite{Maclure1991}. The design is therefore based on retrospective sampling around outcome events, rather than anchoring its estimation on the exposure itself \parencite{Cadarette2021}. Under the assumption that the baseline exposure probability remains constant across the sampled periods, the resulting odds ratio estimates the causal risk ratio associated with the exposure. The design is therefore primarily intended for transient exposures and outcomes with an abrupt onset. Multiple control time sampling strategies \parencite{Navidi1998, Bateson1999} and extensions that relax the exposure trend assumption in this design have been proposed.
\par\medskip
One of these extensions is the CTC design. It extends the CCO design by introducing an external control group to account for temporal time trends in exposure prevalence \parencite{Suissa1995, Suissa1998}. Exposure odds in the control group are used to estimate the population-level exposure trend, which is subsequently removed from the crossover estimate. It thus assumes that the exposure time trends is comparable between controls and cases. The CCTC design replaces external controls with future cases sampled from the same source population to relax this assumption \parencite{Wang2011}. Both of these methods do not require explicit modelling of time trends. All time effects simply cancel out by design, under the required assumptions.
\par\medskip
Below we propose a new method that combines the purely exposure centric estimation strategy of the SCCS design with the design based adjustment used in the CTC and CCTC designs.

\section{The Proposed Method}

\subsection{How it Works}

In the proposed method, $m$ 1:1 matched pairs of exposed individuals are created first. Inside a pair, each individual acts as control for the other individual at their respective exposure time. Only pairs where this is possible are used. For example, consider the pairing of individuals $a$ and $b$, where individual $a$ was exposed at $t = 50$ and individual $b$ was exposed at $t = 120$, and suppose that $\tau = 40$. We would then use the time of individual $b$ at $t \in [50, 90]$ as control for the risk period after exposure of individual $a$ at $t \in [50, 90]$ and do the reverse at $t = 120$, when individual $b$ is exposed instead. Figure~\ref{fig::pair_example} depicts this pair graphically. If individual $b$ was exposed at $t = 30$ instead, we would not be able to form this pair, because the risk periods would overlap. Such pairings are said to be \emph{invalid} pairs. Formally, a pair $(a, b)$ with respective exposure times $t_1$ and $t_2$ is valid only if $|t_1 - t_2| > \tau$ and if the durations $[t_1, t_1 + \tau]$ and $[t_2, t_2 + \tau]$ are fully observed for both $a$ and $b$. These requirements may be relaxed, as discussed in appendix~\ref{appendix::valid_pairs}.

\begin{figure}[!htb]
	\centering
	\includegraphics[width=1\linewidth]{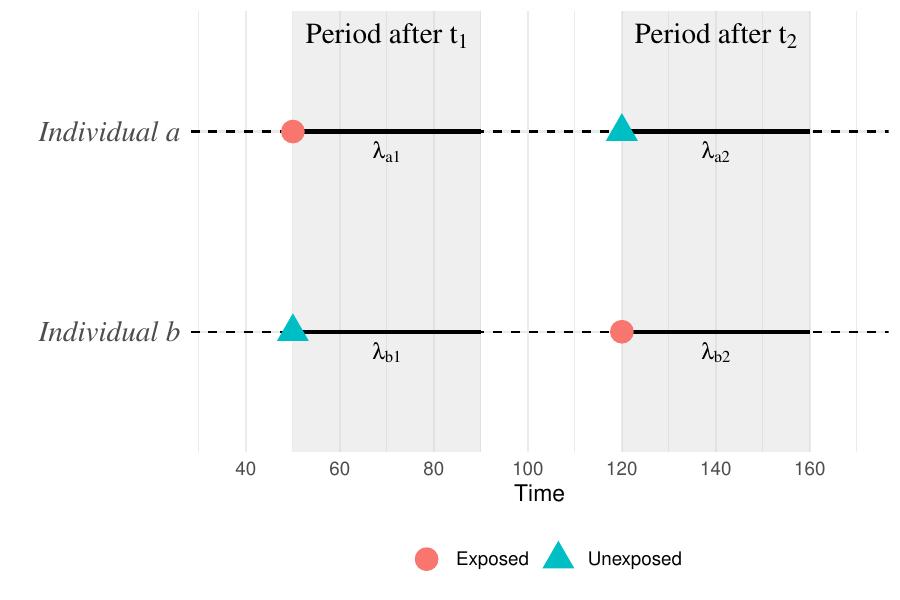}
	\caption{A graphical depiction of an exemplary pair $(a, b)$ in the proposed symmetric pair matching design. Here, individual $a$ was exposed at $t_1 = 50$ and individual $b$ was exposed at $t_2 = 120$ with a risk period duration of $\tau = 40$. $\lambda_{a1}$, $\lambda_{a2}$, $\lambda_{b1}$ and $\lambda_{b2}$ denote the corresponding individual and time period specific event rates.}
	\label{fig::pair_example}
\end{figure}

\par\medskip
The key idea is that the symmetric usage of observation time under control and exposure for both individuals allows perfect cancellation of all time-invariant individual level effects and all time related effects, under some assumptions, similar to the CTC \parencite{Suissa1995} and CCTC \parencite{Wang2011} designs. Suppose that the occurrence of events can be described as a non-homogeneous Poisson process \parencite{CifuentesAmado2015}, so that the number of events experienced by individual $i$ in the $\tau$ duration after $t$ follows a Poisson distribution, where the time period and person-specific event rate is defined as:

\begin{equation} \label{eq::poisson_dist1}
	\ln(\lambda_{it\tau}) = \delta_{i} + \gamma_{t\tau} + A_i(t)\theta.
\end{equation}

Here, $\delta_i$ denotes time-invariant individual-level effects, $\gamma_{t\tau}$ is the time effect of the period and $\theta$ is the fixed exposure effect. Consider a valid matched pair consisting of individuals $a$ and $b$. Let $t_1$ and $t_2$ denote the exposure time of these individuals, with $t_1 < t_2$. The corresponding Poisson event rates are:

\begin{equation} \label{eq::pair_rates}
	\begin{aligned}
		\ln(\lambda_{a,t_1,\tau}) = & \delta_{a} + \gamma_{t_1\tau}, \\
		\ln(\lambda_{b,t_1,\tau}) = & \delta_{b} + \gamma_{t_1\tau} + \theta, \\
		\ln(\lambda_{b,t_2,\tau}) = & \delta_{b} + \gamma_{t_2\tau}, \\
		\ln(\lambda_{a,t_2,\tau}) = & \delta_{a} + \gamma_{t_2\tau} + \theta. \\
	\end{aligned}
\end{equation}

For notational convenience, we will use $\lambda_{a1}$, $\lambda_{b1}$, $\lambda_{b2}$ and $\lambda_{a2}$ instead, respectively. Using some simple algebra these equations can be re-arranged (see appendix~\ref{appendix::theta}) to obtain:

\begin{equation} \label{eq::theta_identity}
	\theta = \frac{1}{2} \ln\! \left(\frac{\lambda_{b1} \lambda_{a2}}{\lambda_{a1} \lambda_{b2}}\right).
\end{equation}

This result is, however, insufficient to obtain an estimate of $\theta$, since we do not actually observe the event rates directly. Instead, we only observe random draws of events that follow a Poisson distribution with the respective event rates. Formally, we assume that:

\begin{equation}
	X_{itj} \sim \text{Poisson}(\lambda_{it\tau}) \quad \forall i,t.
\end{equation}

For each pair, we thus observe four event counts, corresponding to one random draw from the respective time-dependent and individual event rates. Suppose that these counts are contained in four random variables denoted $X_{a1j}$, $X_{a2j}$, $X_{b1j}$ and $X_{b2j}$, where $j \in \{1, ..., m\}$ indexes the $m$ pairs used. Because the event counts of different individuals are assumed to be independent, it follows that:

\begin{equation}
	E(X_{b1j} X_{a2j}) = \lambda_{b1}\lambda_{a2} \quad \text{and} \quad E(X_{a1j} X_{b2j}) = \lambda_{a1}\lambda_{b2},
\end{equation}

for each pair. Using the structural identity from equation~\ref{eq::theta_identity}, we can thus write:

\begin{equation}
	\theta = \frac{1}{2} \ln\! \left(\frac{E(X_{b1j} X_{a2j})}{E(X_{a1j} X_{b2j})}\right),
\end{equation}

which motivates the empirical moments based estimator:

\begin{equation}
	\hat{\theta}_m = \frac{1}{2} \ln\! \left(\frac{\frac{1}{m}\sum_{j = 1}^{m} X_{b1j} X_{a2j}}{\frac{1}{m}\sum_{j = 1}^{m} X_{a1j} X_{b2j}}\right).
\end{equation}

In words, we take the average over the individual-level event products as an estimate of the average over the expectations of individual-level products. Due to its use of individual-level multiplication, only individuals that have experienced at least one event and were exposed at some point in time contribute to the final estimate, similar to the original SCCS design \parencite{Farrington1995}. Appendix~\ref{appendix::consistency} contains a more thorough motivation and mathematical proof that this estimator is asymptotically unbiased. The associated variance and confidence interval may be calculated using individual-level bootstrapping \parencite{Efron1986}, as described in more detail in appendix~\ref{appendix::bootstrap}.

\subsection{Forming the Pairs} \label{chap::pairs}

This of course raises the question how the pairs should be created. One might be tempted to use each individual in only a single pair. This method is inefficient, because it uses only $2\tau$ time units per person and thus likely disregards large parts of the observed data. To preserve more of the available information, we may form \emph{all} possible valid pairs instead. Although this approach uses individuals multiple times, the estimator stays asymptotically unbiased, as long as the number of valid pairs grows faster than linearly with $n$. This is guaranteed to be true if the probability of a valid pair is positive (see appendix~\ref{appendix::consistency}). For compuational efficiency, users may use a large number of randomly choosen pairs instead, without much loss of information. All of these options are implemented in the \texttt{SPMD} R package.

\subsection{Required Assumptions}

Multiple assumptions have to be met for the application of this estimator. First, the four causal identifiability assumptions stated in section~\ref{chap::target_estimand} have to be met. Since the presented method is a self-controlled method, the exchangeability assumption holds by design with respect to time-constant confounders, but may be violated by the presence of time-dependent confounders. Importantly, the effect of $A_i(t)$ on the outcome has to be transient, meaning that it only has an effect in the $\tau$ time units after exposure, where $\tau$ is the same for all individuals and known. This exposure effect also has to be the same for all individuals. The outcome must be recurrent, so that events do not preclude further events. Additionally, the exposure must be independent of previous events. These assumptions are also required for the SCCS design \parencite{Farrington2018, Etievant2026}. Unlike in the standard SCCS, both the exposure and event baseline probabilities may be arbitrary functions of time. However, the effect of time-constant confounders and the exposure must not vary over time, and the outcome time trend must be the same for all indiviuals. The latter two assumptions are similar to the ones required for the CTC and CCTC designs \parencite{Suissa1995, Suissa1998, Wang2011}.
\par\medskip
Right-censoring may occur, as long as it is not caused by time-dependent factors that also cause the exposure or outcome. This includes the outcome itself, e.g. outcome dependent censoring may not be present. The only parametric assumption made by the method is the Poisson assumption described in equation~\ref{eq::poisson_dist1}. No further modeling assumptions are required. As in other self-controlled designs, not all of these assumptions are empirically testable. Researchers have to carefully assess all assumptions using subject-matter knowledge on a case specific basis, before drawing causal conclusions.

\section{Simulation Study} \label{chap::simulation}

We performed a small Monte-Carlo simulation study to investigate the finite sample performance of the proposed method in comparison to existing ones. This simulation study is designed similarly to the pre-registered study by \textcite{Meiszl2026a}. The employed data generation process (DGP) contains three variables: the time-dependent binary exposure status $A_i(t)$, the recurrent outcome $Y_i(t)$, which is 1 whenever an event occurs and 0 otherwise, and a standard normally distributed time-fixed baseline confounder $U_i$. The data is generated for the time period $[0, 1000]$. The first and only exposure time per individual $T_{A_i}$ is generated from a Cox proportional hazards model \parencite{Cox1972} of the form:

\begin{equation}
	h_A(t | U_i) = h_{A0}(t) \exp(U_i\ln(2)).
\end{equation}

Based on the generated exposure time, we then define an exposure history variable $E_i(t)$ as:

\begin{equation}
	E_i(t) = I(T_{A_i} < t \leq T_{A_i} + \tau).
\end{equation}

The recurrent outcome events are generated from the following Cox proportional hazards model:

\begin{equation}
	h_Y(t | U_i, E_i(t)) = h_{Y0}(t) \exp(U_i\ln(2) + E_i(t) \ln(2.5)).
\end{equation}

Here, $h_{A0}(t)$ denotes the exposure baseline hazard and $h_{Y0}(t)$ the outcome baseline hazard. In scenario 1, both of these are constant over time, whereas in scenario 2 the exposure and outcome time probability oscillate out of phase from each other, as shown in figure~\ref{fig::time_trends}. This latter type of time trend pattern is often observed in influenza vaccine safety studies, where vaccination campaigns occur before the influenza season each year. Because both the probability for the exposure and the outcome are time-dependent, confounding based on time is induced. The only other confounder in both scenarios is $U_i$, because it causes both $A_i(t)$ and $Y_i(t)$. Note that because previous events in $Y_i(t)$ do not influence future events, $Y_i(t)$ is mathematically equivalent to a non-homogeneous Poisson process \parencite{CifuentesAmado2015}. No censoring beyond administrative censoring at study end was included.

\begin{figure}[!htb]
	\centering
	\includegraphics[width=1\linewidth]{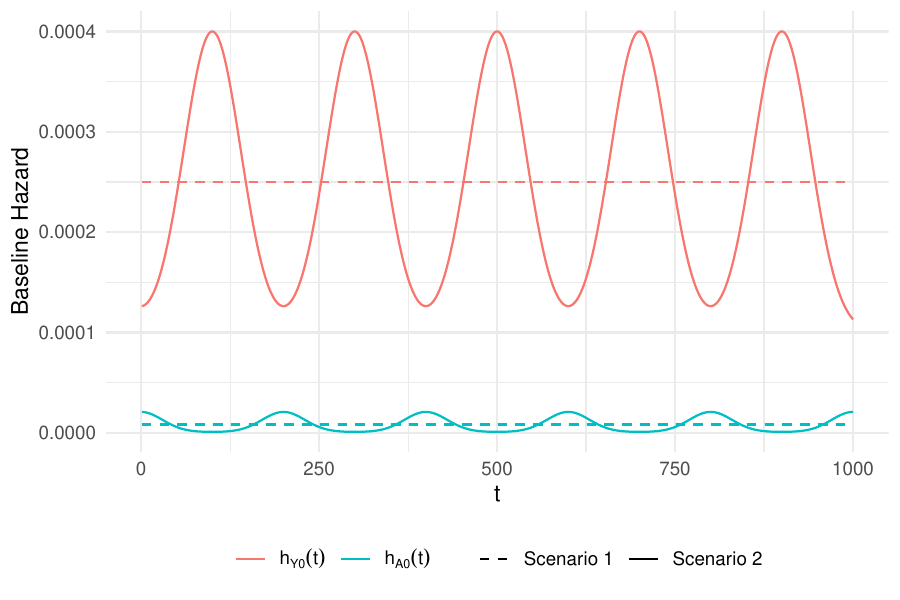}
	\caption{A graphical depiction of the baseline exposure hazard ($h_{A0}(t)$) and baseline outcome hazard ($h_{Y0}(t)$) as a function of time in both simulation scenarios.}
	\label{fig::time_trends}
\end{figure}

\par\medskip
For both scenarios, we generated $1000$ datasets with $n = 20000$ individuals each using a discrete-event simulation approach \parencite{Gillespie1977, Denz2026}. Afterwards, we applied the standard SCCS \parencite{Farrington1995}, the natural spline based SCCS \parencite{GhebremichaelWeldeselassie2017} with 5 and 15 knots placed at suitably choosen quantiles, the standard CCO with a control period directly preceeding the risk period \parencite{Maclure1991}, the CTC method \parencite{Suissa1995} and the proposed estimator using all valid pairs to each dataset. Figure~\ref{fig::sim_results} shows boxplots of the difference between $\theta$ and $\hat{\theta}$ for all included methods and scenarios.

\begin{figure}[!htb]
	\centering
	\includegraphics[width=1\linewidth]{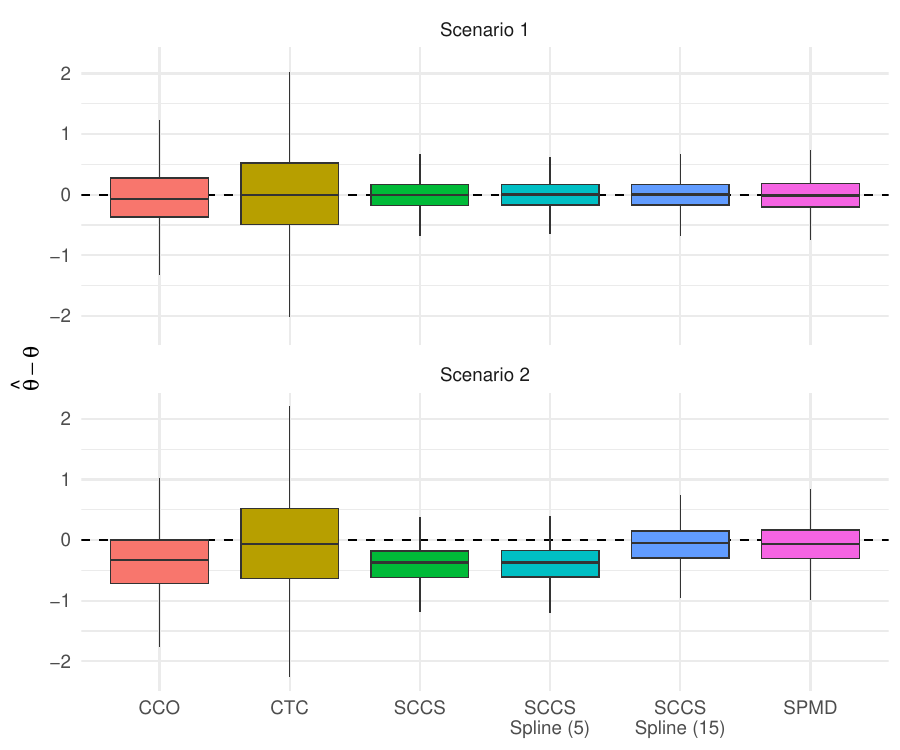}
	\caption{Boxplots of the deviations between the estimated and true $\theta$ for all considered methods and simulation scenarios. A sample size of $n = 20000$ was used with $1000$ simulation replications. Outliers are not shown.}
	\label{fig::sim_results}
\end{figure}

\par\medskip
In scenario 1 without any time trends, all methods produced unbiased estimates on average. The variance of the estimates is comparable between all three SCCS variants and the proposed method, whereas the CCO and CTC estimates have a much larger spread around the true value, because both method uses only a small fraction of the available information. In scenario 2, the standard SCCS and CCO methods showed substantial bias, due to the time-induced confounding. Similarly, the spline based SCCS with only 5 knots was not able to correctly model the time trends and thus also produced biased results. With 15 knots on the other hand, the spline based SCCS method was unbiased. The CTC method also managed to adjust for the time effects, albeit with a larger variance than the spline based SCCS. Finally, the proposed method also produced unbiased estimates, with an efficiency comparable to the correctly specified spline based SCCS. The average bias and mean squared error and associated Monte-Carlo standard errors are given in appendix~\ref{appendix::simulation}. This appendix also shows further simulation results, in which the sample size and $\theta$ were varied.

\newpage

\section{Discussion}

We introduced SPM, a self-controlled study design that combines automatic control of time-invariant confounding with design-based adjustment for temporal effects. Like other self-controlled designs such as the SCCS \parencite{Farrington1995}, CCO \parencite{Maclure1991} or CTC \parencite{Suissa1995} designs, SPM makes within-individual comparisons and therefore controls for confounding by characteristics that remain constant during the observation period. In addition, its underlying symmetric pairs account for temporal effects without modeling them explicitly. These pairs are built only from exposed individuals, requiring no external control group that may induce selection bias \parencite{Suissa1995, Wang2011}. As such, this design may be particularly useful in pharmacoepidemiology and vaccine safety studies, where both exposure and outcome incidence may vary over time \parencite{Hallas2014, Nie2022, Bots2025a}. It would, conversely, not be appropriate for the analysis of long-term or sustained treatment effects.
\par\medskip
An important feature of SPM is that, contrary to self-controlled methods anchored on the outcome, it uses information contained in the observation period after event occurrences. When appropriate, this greatly increases efficiency as compared to outcome anchored designs \parencite{Takeuchi2018}. In our simulations, SPM therefore showed greater efficiency than the alternative design-based approaches considered \parencite{Maclure1991, Suissa1995}, while remaining unbiased in the presence of temporal trends in both exposure and outcome incidence. The estimator is also deliberately simple: it is based on sums of events in the matched time periods and thus does not require fitting a complex statistical model. This makes the method straightforward to implement and interpret, and avoids the need to specify a functional form for temporal trends or seasonality. Unlike other recently proposed self-controlled methods \parencite{Allison2006, Gasparrini2021}, it can also be directly applied in continuous time, without categorization of the underlying time scale.
\par\medskip
SPM nevertheless retains the fairly restrictive assumptions inherent to self-controlled study designs \parencite{Farrington1995, Farrington2018, Etievant2026}. It should therefore not be viewed as a general solution to causal inference problems, but as an alternative to the SCCS and related methods when their underlying assumptions are plausible and temporal effects are an important concern. If some of the required assumptions, such that events must be independent of past events, are violated, existing extensions for the SCCS relaxing these assumptions are likely preferable \parencite{Farrington2011}. Future work should investigate similar extensions to the proposed SPM for settings involving, for example, non-independent events, terminal events, or violations of the parallel trends assumption. Currently, if the occurrence of an event changes subsequent exposure patterns of an individual, case-crossover based methods are better suited for causal effect estimation.
\par\medskip
A further limitation is that SPM, like other self-controlled designs, currently does not support adjustment for time-dependent confounders. In settings involving complex treatment processes and time-dependent confounding, methods such as marginal structural models \parencite{Robins2000} or time-dependent matching \parencite{Thomas2020} may therefore be more appropriate. SPM can also become computationally demanding when the number of eligible pairings is large, although the accompanying R package alleviates those issues by offering an efficient implementation relying on state-of-the-art software backend for all required steps \parencite{Barrett2026}.
\par\medskip
In conclusion, SPM provides a simple design-based approach to control for temporal effects in self-controlled studies. By combining the efficiency of the SCCS, gained through inclusion of observation time regardless of event occurences, with adjustment for temporal effects that does not require explicit modeling, SPM offers a practical alternative for studies of transient exposures, including vaccine safety and drug safety studies. It is not intended to replace methods designed for more complex causal structures, but provides an efficient and transparent option when the assumptions of a self-controlled design are appropriate.

\section*{Acknowledgments}

We would like to thank all members of the Department of Medical Informatics, Biometry and Epidemiology at the Ruhr-University Bochum for their valuable input and the multiple discussions about the contents of this paper.

\section*{Financial disclosure}

None reported.

\section*{Conflict of interest}

The authors declare no potential conflict of interests.

\FloatBarrier
\newpage

\printbibliography

\appendix

\newpage

\counterwithin{equation}{section}
\counterwithin{figure}{section}
\counterwithin{table}{section}

\section{Derivation of the Equation for $\theta$} \label{appendix::theta}

To obtain equation~\ref{eq::theta_identity} given the four pair-level rates as defined in equation~\ref{eq::pair_rates}, we first substract within each time point, which yields:

\begin{equation}
	\ln(\lambda_{b1}) - \ln(\lambda_{a1}) = (\delta_{b} - \delta_{a}) + \theta,
\end{equation}

and:

\begin{equation}
	\ln(\lambda_{a2}) - \ln(\lambda_{b2}) = (\delta_{a} - \delta_{b}) + \theta.
\end{equation}

Adding these two equations together we obtain:

\begin{equation}
	(\ln(\lambda_{b1}) - \ln(\lambda_{a1})) + (\ln(\lambda_{a2}) - \ln(\lambda_{b2})) = 2\theta,
\end{equation}

which can be simplified to:

\begin{equation}
	\theta = \frac{1}{2} \Big[\ln(\lambda_{b1}) - \ln(\lambda_{a1}) + \ln(\lambda_{a2}) - \ln(\lambda_{b2}) \Big].
\end{equation}

This equation may be re-written as:

\begin{equation}
	\theta = \frac{1}{2} \ln\!\left(\frac{\lambda_{b1} \lambda_{a2}}{\lambda_{a1} \lambda_{b2}}\right).
\end{equation}

This cancellation of individual-level and time effects does theoretically not require the assumption that each rate is the parameter of a Poisson distributed. A similar derivation could be made for any model in which the individual-level effects and time effects have a multiplicative effect on the expected value of the event count distribution. We choose to focus on the Poisson assumption here, because it made the derivation of the estimator and the associated proof of consistency in appendix~\ref{appendix::consistency} more convenient. See appendix~\ref{appendix::consistency_poisson} for a small discussion.

\newpage

\section{Consistency of the Proposed Estimator} \label{appendix::consistency}

\subsection{Motivation and Notation}

We first establish the form of the empirical moments based estimator. For notational convenience, we adopt a slightly different notation as compared to the main text. Suppose we observe data on $n$ independent individuals, indexed by $1, ..., n$. For individual $i$, let $N_i(t)$ denote the cumulative number of events observed up to time $t$. Thus:

\begin{equation}
	N_i(t_2) - N_i(t_1)
\end{equation}

is the number of events occurring during the time period $[t_1, t_2]$. Let $T_i$ denote the exposure time of individual $i$ and define:

\begin{equation}
	N_{ij} = N_i(T_j + \tau) - N_i(T_j),
\end{equation}

which is the number of events experienced by individual $i$ during the risk period of individual $j$. We condition throughout on the exposure times $T_1, ..., T_n$ and assume that whenever $i = j$ or when the two exposure windows do not overlap, e.g. when:

\begin{equation}
	[T_i, T_i + \tau] \cap [T_j, T_j + \tau] = \emptyset,
\end{equation}

the random variable $N_{ij}$ follows a Poisson distribution, with:

\begin{equation}
	N_{ij} \sim \text{Poisson}(\lambda_{ij\tau}),
\end{equation}

where:

\begin{equation}
	\lambda_{ij\tau} = \exp(\delta_i + \gamma_{T_j\tau} + A_i(T_j)\theta).
\end{equation}

Here, $A_i(T_j)$ indicates whether individual $i$ was exposed at $T_j$, $\delta_i$ is the time constant individual-level effect and $\gamma_{T_j\tau}$ is the time effect during the interval $[T_j, T_j + \tau]$. Using this notation, the event rates inside a pair $(i, j)$ are defined as:

\begin{equation} \label{eq::pair_rates2}
	\begin{aligned}
		\lambda_{ij\tau} = & \exp(\delta_{i} + \gamma_{T_j\tau}), \\
		\lambda_{jj\tau} = & \exp(\delta_{j} + \gamma_{T_j\tau} + \theta), \\
		\lambda_{ji\tau} = & \exp(\delta_{j} + \gamma_{T_i\tau}), \\
		\lambda_{ii\tau} = & \exp(\delta_{i} + \gamma_{T_i\tau} + \theta). \\
	\end{aligned}
\end{equation}

From the derivation in appendix~\ref{appendix::theta} we know that:

\begin{equation} \label{eq::theta2}
	\lambda_{jj\tau}\lambda_{ii\tau} = e^{2\theta} \lambda_{ij\tau}\lambda_{ji\tau},
\end{equation}

for all valid pairs. A pair $(i, j)$ is defined to be valid if and only if $i \neq j$ and $|T_i - T_j| > \tau$. We thus define the set of all valid ordered pairs as:

\begin{equation}
	E_n = \{(i, j) : [T_i, T_i + \tau] \cap [T_j, T_j + \tau] = \emptyset\}.
\end{equation}

We use ordered pairs throughout the proof purely for notational simplicity. Since the average over unordered pairs is equal to the average over ordered pairs, this has no impact on the arguments. For additional notational simplicity, we define:

\begin{equation}
	X_{ij} = N_{ij} N_{ji}, \quad \text{and} \quad Y_{ij} = N_{ii} N_{jj}.
\end{equation}

Their empirical averages are defined as:

\begin{equation}
	\bar{X}_n = \frac{1}{|E_n|} \sum_{(i, j) \in E_n} X_{ij},
\end{equation}

and:

\begin{equation}
	\bar{Y}_n = \frac{1}{|E_n|} \sum_{(i, j) \in E_n} Y_{ij}.
\end{equation}

Since independent Poisson variables satisfy:

\begin{equation}
	E(N_{ij} N_{ji}) = \lambda_{ij} \lambda_{ji},
\end{equation}

we obtain:

\begin{equation}
	E(\bar{Y}_n) = e^{2\theta} E(\bar{X}_n).
\end{equation}

This motivates the estimator:

\begin{equation}
	\hat{\theta}_n = \frac{1}{2} \ln\left(\frac{\bar{Y}_n}{\bar{X}_n}\right).
\end{equation}

The following theorem establishes consistency of this estimator even when the valid pairs are dependent, such that an individual may appear in more than one pair.

\subsection{Consistency}

\begin{theorem}[Consistency of the SPM estimator]
	\label{thm::consistency}
	Condition on the exposure times $T_1,\ldots,T_n$. Suppose that, for every relevant $i,j,n$, the event counts are Poisson distributed with means $\lambda_{ij\tau}$ satisfying
	\begin{equation}
		0 < c \leq \lambda_{ij\tau} \leq C < \infty
	\end{equation}
	for fixed constants $c$ and $C$. Suppose further that, for every valid pair $(i,j)\in E_n$, the counts $N_{ij}$ and $N_{ji}$ are independent. Let
	\begin{equation}
		A_n
		=
		{((i,j),(k,l))\in E_n^2:
			|{i,j,k,l}|=3}.
	\end{equation}
	If the very mild conditions (see~\ref{appendix::consistency_conditions})
	\begin{equation}
		\frac{1}{|E_n|}\longrightarrow 0
		\qquad\text{and}\qquad
		\frac{|A_n|}{|E_n|^2}\longrightarrow 0,
	\end{equation}
	are satisfied, then
	\begin{equation}
		\hat{\theta}_n-\theta
		\overset{p}{\longrightarrow}0.
	\end{equation}
\end{theorem}

\begin{proof}
	
	We first establish concentration of $\bar{X}_n$ around its expectation. Since products of Poisson variables have moments of all orders that are bounded uniformly when their means are uniformly bounded, there exists a finite constant $M_X$ such that:
	
	\begin{equation}
		\underset{i,j}{\sup}\,\text{Var}(X_{ij}) \leq M_X < \infty,
	\end{equation}
	
	and analogously:
	
	\begin{equation}
		\underset{i,j}{\sup}\,\text{Var}(Y_{ij}) \leq M_Y < \infty.
	\end{equation}
	
	Moreover,
	
	\begin{equation} \label{eq::X_ij_bound}
		E(X_{ij}) = \lambda_{ij\tau} \lambda_{ji\tau} \geq c^2
	\end{equation}
	
	and
	
	\begin{equation}
		E(Y_{ij}) = \lambda_{ii\tau} \lambda_{jj\tau} \geq c^2.
	\end{equation}
	
	Therefore:
	
	\begin{equation}
		E(\bar{X}_n) \geq c^2 \quad \text{and} \quad E(\bar{Y}_n) \geq c^2.
	\end{equation}
	
	Consider:
	
	\begin{equation}
		\text{Var}(\bar{X}_n) = \frac{1}{|E_n|^2} \sum_{(i, j), (k, l) \in E_n} \text{Cov}(X_{ij}, X_{kl}).
	\end{equation}
	
	We partition the terms according to the number of distinct individuals among $i, j, k, l$. First, consider terms for which:
	
	\begin{equation}
		|\{i, j, k , l\}| = 2.
	\end{equation}
	
	Because $i \neq j$ and $k \neq l$ these are precisely cases in which:
	
	\begin{equation}
		(k, l) = (i, j) \quad \text{or} \quad (k, l) = (j, i).
	\end{equation}
	
	Note that:
	
	\begin{equation}
		X_{ij} = N_{ij} N_{ji} = X_{ji}.
	\end{equation}
	
	Thus the terms are not independent, but there are only $2 |E_n|$ such ordered combinations if the two representations are counted separately. Hence their total contribution is bounded in absolute value by:
	
	\begin{equation}
		\frac{2 |E_n| M_X}{|E_n|^2} = \frac{2M_X}{|E_n|}.
	\end{equation}
	
	Consequently, this contribution converges to zero whenever:
	
	\begin{equation}
		\frac{1}{|E_n|} \rightarrow 0.
	\end{equation}
	
	Next consider terms for which:
	
	\begin{equation}
		|\{i, j, k, l\}| = 4.
	\end{equation}
	
	In this case, the two pairs contain disjoint sets of individuals. Since the individuals are independent, all random variables entering $X_{ij}$ are independent of those entering $X_{kl}$. Therefore:
	
	\begin{equation}
		\text{Cov}(X_{ij}, X_{kl}) = 0.
	\end{equation}
	
	Finally, consider terms for which:
	
	\begin{equation}
		|\{i, j, k, l\}| = 3.
	\end{equation}
	
	These correspond to two valid pairs sharing exactly one individual. Using the Cauchy-Schwarz inequality, it follows that:
	
	\begin{equation}
		|\text{Cov}(X_{ij}, X_{kl})| \leq \sqrt{\text{Var}(X_{ij}) \text{Var}(X_{kl})} \leq M_X.
	\end{equation}
	
	Consequently:
	
	\begin{equation}
		\left| \frac{1}{|E_n|^2} \sum_{((i, j), (k, l)) \in A_n} \text{Cov}(X_{ij}, X_{kl}) \right| \leq M_X \frac{|A_n|}{|E_n|^2}.
	\end{equation}
	
	We therefore obtain the bound:
	
	\begin{equation}
		\text{Var}(\bar{X}_n) \leq \frac{2M_X}{|E_n|} + M_X \frac{|A_n|}{|E_n|^2}.
	\end{equation}
	
	Thus:
	
	\begin{equation}
		\text{Var}(\bar{X}_n) \longrightarrow 0.
	\end{equation}
	
	provided that:
	
	\begin{equation}
		\frac{1}{|E_n|} \longrightarrow 0,
	\end{equation}
	
	and:
	
	\begin{equation}
		\frac{|A_n|}{|E_n|^2} \longrightarrow 0.
	\end{equation}
	
	The same argument holds for $\bar{Y}_n$. Next, we use this result to show that $\ln(\bar{X}_n)-\ln(E(\bar{X}_n))\to 0$ in probability: Fix $\epsilon >0$ and set
	
	\begin{equation}
		D_n = \{\epsilon <|\ln(\bar{X}_n) - E(\ln(\bar{X}_n))|\}.
	\end{equation}
	
	By equation~\ref{eq::X_ij_bound},
	
	\begin{equation}
		0<\ell\coloneqq \inf_n E(\bar{X}_n)
	\end{equation}
	
	and we set
	
	\begin{equation}
		B_n = \{|\bar{X}_n - E(\bar{X}_n)| < \ell/2\}.
	\end{equation}
	
	Then $P(B_n^c) \to 0$, and on $B_n$,
	
	\begin{equation}
		\bar{X}_n > E(\bar{X}_n) - \ell/2 \ge \ell - \frac{\ell}{2} = \frac{\ell}{2}.
	\end{equation}
	
	Also, $E(\bar{X}_n) \ge \ell > \ell/2$ trivially. So on $B_n$, both $\bar{X}_n$ and $E(\bar{X}_n)$ lie in $[\ell/2, \infty)$. Furthermore, $\ln|_{[\ell/2,\infty)}$ is Lipschitz continuous with Lipschitz constant $2/\ell$: Indeed,
	
	\begin{equation}
		\forall x\geq \ell/2:\ln'(x) = \frac{1}{x} \le \frac{2}{\ell} 
	\end{equation}
	
	and thus, by the Newton–Leibniz theorem,
	
	\begin{equation}
		\forall x_1,x_2\geq \ell/2:|\ln(x_1) - \ln(x_2)|=\int_{\min(x_1,x_2)}^{\max(x_1,x_2)}\frac{1}{x}dx \le \frac{2}{\ell}|x_1-x_2|.
	\end{equation}
	
	Consequently,
	
	\begin{equation}
		\text{on } B_n: \qquad |\ln(\bar{X}_n) - \ln(E(\bar{X}_n))| \le (2/\ell)|\bar{X}_n - E(\bar{X}_n)|
	\end{equation}
	
	and thus
	
	\begin{equation}
		\begin{aligned}
			P(\epsilon <|\ln(\bar{X}_n) - \ln(E(\bar{X}_n))| ) & =P(D_n)=P(D_n\cap B_n)+ P(D_n\cap B_n^c)\\
			&\leq P(D_n\cap B_n)+P(B_n^c)\\
			&\leq P(|\bar{X}_n - E(\bar{X}_n)| > \epsilon\ell/2)+P(B_n^c) \\
		\end{aligned}
	\end{equation}
	
	and both terms $\to 0$. Putting everything together yields
	
	\begin{equation}
		\begin{split}
			2|\hat\theta_n-\theta|=\left|\ln(\bar{Y}_n/\bar{X}_n)-\ln(E(\bar{Y}_n)/E(\bar{X}_n))\right|\leq \quad \quad \quad \quad \quad \quad \\  |\ln(\bar{Y}_n)-\ln(E(\bar{Y}_n))|+ |\ln(\bar{X}_n)-\ln(E(\bar{X}_n))|\to 0.
		\end{split}
	\end{equation}
	
	Therefore,
	
	\begin{equation}
		\hat{\theta}_n = \frac{1}{2} \ln\left(\frac{\bar{Y}_n}{\bar{X}_n}\right) \overset{p}{\longrightarrow} \frac{1}{2} \ln\left(e^{2\theta}\right) = \theta.
	\end{equation}
	
	Hence,
	
	\begin{equation}
		\hat{\theta}_n - \theta \overset{p}{\longrightarrow} 0,
	\end{equation}
	
	establishing consistency of the estimator under the condition that the event rates are uniformly bounded and if the conditions:
	
	\begin{equation}
		\frac{1}{|E_n|} \longrightarrow 0 \quad \text{and} \quad \frac{|A_n|}{|E_n|^2} \longrightarrow 0
	\end{equation}
	
	hold.
	
\end{proof}

\subsection{Conditions for Consistency} \label{appendix::consistency_conditions}

The two conditions above have a simple interpretation. The first requires the number of valid pairs to diverge. The second requires the proportion of dependent pair-pair comparisons to become asymptotically negligible.
\par\medskip
Below we use two extreme scenarios to illustrate when consistency holds. First, suppose that instead of using all valid pairs, we only use each individual in a single pair. In this case $|E_n| \rightarrow \infty$, so the first condition holds. Furthermore, $|A_n| = 0$, so the fraction $\frac{|A_n|}{|E_n|^2}$ is 0 by definition. Consistency thus holds. Similarly, suppose that we use all possible valid ordered pairs and all of them are valid. In this case:

\begin{equation}
	|E_n| = n (n - 1).
\end{equation}

For a fixed ordered pair $(i, j)$, another ordered pair $(k, l)$ shares exactly one individual with $(i, j)$ in at most four ways:

\begin{equation}
	k = i, \quad l = i, \quad k = j, \quad \text{or} \quad l = j.
\end{equation}

Once the shared individual and its position have been specified, there are at most $n - 2$ choices for the remaining individual. Hence:

\begin{equation}
	|A_n| \leq 4 (n - 2) |E_n|.
\end{equation}

Therefore:

\begin{equation}
	\frac{|A_n|}{|E_n|^2} \leq \frac{4 (n - 2)}{|E_n|}.
\end{equation}

It follows that consistency holds in this case as well. More interestingly, this result also implies that a sufficient condition for the dependence term to vanish is:

\begin{equation}
	\frac{n}{|E_n|} \longrightarrow 0.
\end{equation}

The condition also implies $\frac{1}{|E_n|} \rightarrow 0$, since $n \geq 1$. Hence, the single condition:

\begin{equation}
	\frac{n}{|E_n|} \longrightarrow 0
\end{equation}

is sufficient for consistency under the above uniform boundedness assumptions. This means that whenever the number of valid pairs grows faster than linearly in $n$, consistency follows automatically. For example, suppose that the exposure times $T_1, ..., T_n$ are independent realizations from a fixed distribution and that:

\begin{equation}
	p = P(|T_1 - T_2| > \tau) > 0.
\end{equation}

Because we consider ordered pairs, each of the (n(n-1)) possible ordered pairs is valid with probability $p$. Thus, the expected number of valid ordered pairs is

\begin{equation}
	E(|E_n|)=pn(n-1).
\end{equation}

Consequently, as $n$ increases, the number of valid pairs grows on the order of $n^2$, and we may write

\begin{equation}
	|E_n| \approx p n(n-1).
\end{equation}

It follows that

\begin{equation}
	\frac{|E_n|}{n} \approx p(n-1) \longrightarrow \infty,
\end{equation}

or, equivalently,

\begin{equation}
	\frac{n}{|E_n|} \longrightarrow 0.
\end{equation}

Therefore,

\begin{equation}
	\frac{|A_n|}{|E_n|^2} \leq \frac{4(n-2)}{|E_n|} \longrightarrow 0.
\end{equation}

This means that as long as there is a positive probability that valid pairs can be formed, the condition always holds, regardless of the type of exposure time distribution. The only requirement is that the distribution is fixed and thus independent of $n$.
\par\medskip
It is important to emphasize that these results hold only asymptotically. In any finite sample, there is no guarantee for unbiased estimation. However, the appropriateness of the dependence condition may be checked approximately for any given dataset, since $\frac{|A_n|}{|E_n|^2}$ can be calculated exactly given a dataset of size $n$ with fixed exposure times. Large values indicate that a substantial proportion of pair comparisons are dependent, in which case the theoretical justification for the estimator becomes weak. The associated R package automatically computes this quantity whenever individuals are re-used across pairs.
\par\medskip
It should also be noted that $\frac{|A_n|}{|E_n|^2}$ is a conservative diagnostic rather than a direct measure of the magnitude of the covariance distribution. The variance bound replaces the sum of covariances by the sum of their absolute values and therefore ignores possible cancellation between positive and negative covariance terms. Consequently, the actual covariance contribution can be substantially smaller than this upper bound.

\subsection{The Poisson Assumption} \label{appendix::consistency_poisson}

Throughout both the main text and the proof of consistency above, we assumed that the event counts are Poisson distributed. This assumption is convenient for specifying the data-generating process and deriving the properties of the estimator, but it is stronger than necessary for consistency.
\par\medskip
Inspection of the proof shows that the Poisson assumption is used only to establish two properties. The first is that the relevant cross-moments factorize according to the mean structure. Specifically, for every valid pair $(i, j)$,

\begin{equation}
	E(N_{ij} N_{ji}) = \lambda_{ij\tau} \lambda_{ji\tau},
\end{equation}

and

\begin{equation}
	E(N_{ii} N_{jj}) = \lambda_{ii\tau} \lambda_{jj\tau}.
\end{equation}

These equalities follow from the assumed mean structure together with independence of the corresponding event counts. In particular, the Poisson distribution itself is not required. The second property is that the variances of the products:

\begin{equation}
	X_{ij} = N_{ij} N_{ji} \quad \text{and} \quad Y_{ij} = N_{ii} N_{jj}
\end{equation}

are uniformly bounded. A sufficient condition for this is that the corresponding event counts have uniformly bounded fourth moments. This condition is satisfied, for example, by Poisson or negative binomial event counts when their parameters are uniformly bounded.
\par\medskip
Consequently, the consistency result extends beyond the Poisson setting. Specifically, the same proof establishes consistency for any data-generating process satisfying the required mean structure, the relevant independence conditions, and uniformly bounded variances of $X_{ij}$ and $Y_{ij}$, together with the conditions on the valid-pair structure established above. The Poisson assumption should therefore be regarded as a convenient sufficient assumption rather than a fundamental requirement for consistency.

\newpage

\section{On the Definition of Valid Pairs} \label{appendix::valid_pairs}

In the main manuscript, we defined a pair $(a, b)$ with respective exposure times $t_1$ and $t_2$ to be valid only if $|t_1 - t_2| > \tau$ and if the durations $[t_1, t_1 + \tau]$ and $[t_2, t_2 + \tau]$ were fully observed for both $a$ and $b$. Both of these requirements may be relaxed somewhat, allowing usage of more observation time, which might result in a gain of efficiency of the estimator. These additional considerations were not introduced in the main text due to issues of space and because they make the notation more complex. First we consider right-censoring.

\begin{figure}[!htb]
	\centering
	\includegraphics[width=1\linewidth]{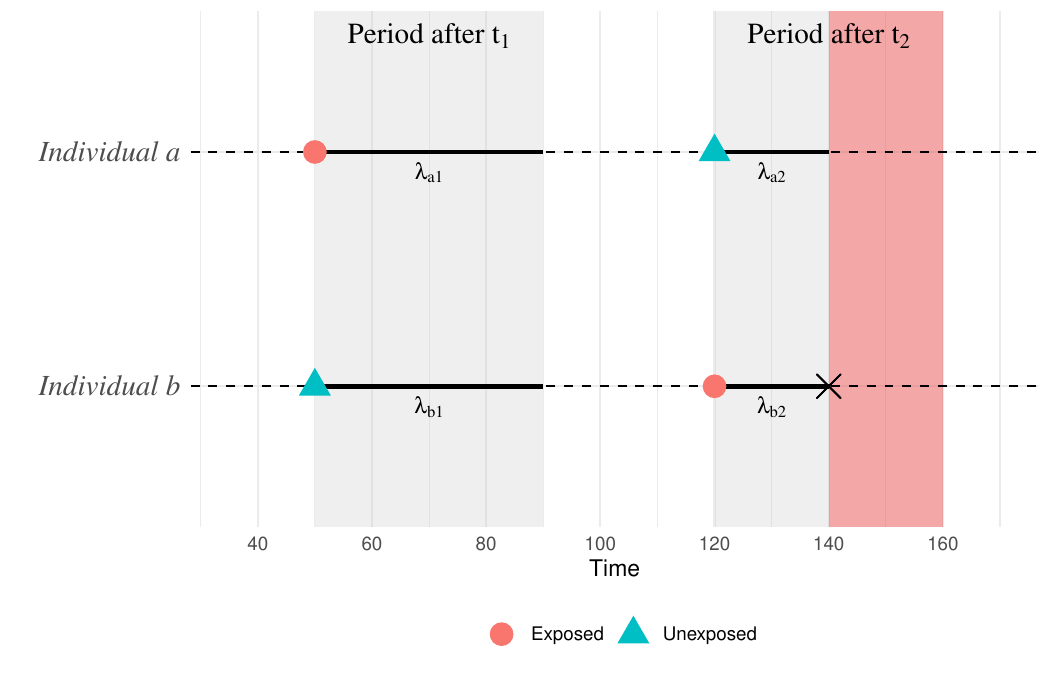}
	\caption{A graphical depiction of an exemplary valid pair $(a, b)$ in the proposed symmetric pair matching design. Here, individual $a$ was exposed at $t_1 = 50$ and individual $b$ was exposed at $t_2 = 120$ with a risk period duration of $\tau = 40$. Additionally, individual $b$ was censored at $t = 140$. $\lambda_{a1}$, $\lambda_{a2}$, $\lambda_{b1}$ and $\lambda_{b2}$ denote the corresponding individual and time period specific event rates.}
	\label{fig::pair_censored_valid}
\end{figure}

\begin{figure}[!htb]
	\centering
	\includegraphics[width=1\linewidth]{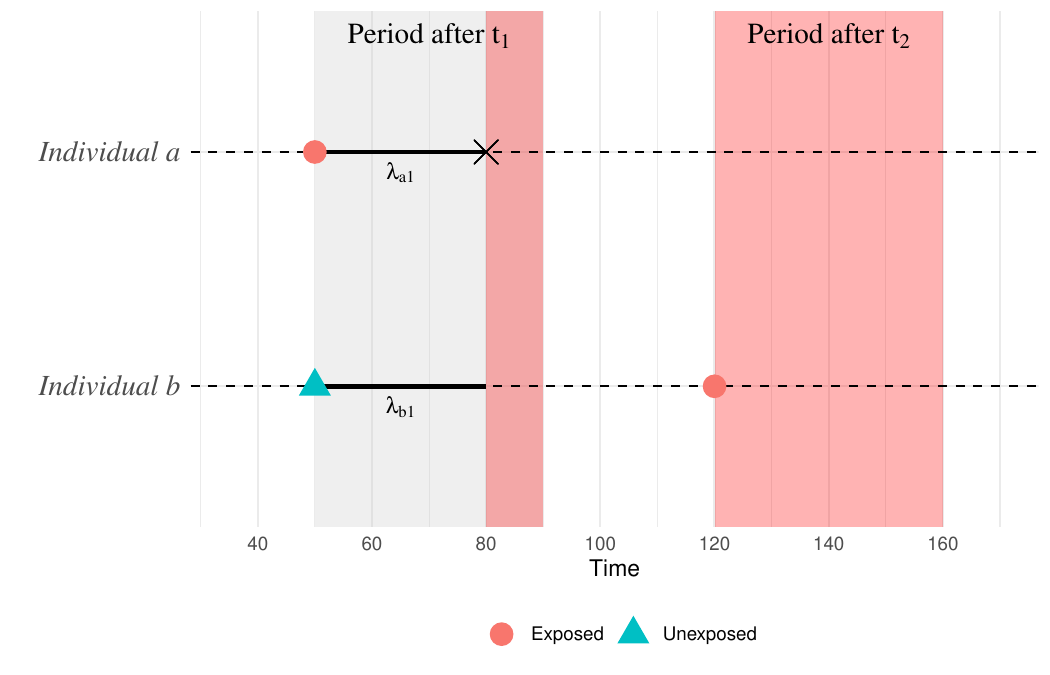}
	\caption{A graphical depiction of an exemplary invalid pair $(a, b)$ in the proposed symmetric pair matching design. Here, individual $a$ was exposed at $t_1 = 50$ and individual $b$ was exposed at $t_2 = 120$ with a risk period duration of $\tau = 40$. Additionally, individual $a$ was censored at $t = 80$, making the pair invalid. $\lambda_{a1}$ and $\lambda_{b1}$ denote the corresponding individual and time period specific event rates.}
	\label{fig::pair_censored_invalid}
\end{figure}

Previously, we only considered a pair to be valid if right-censoring occured after $t_2 + \tau$ for both individuals $a$ and $b$ of a pair. It is, however, possible to use pairs in which right-censoring occurs for one or both pairs after the later of the two exposure times, e.g. $t_2$. Then, we would still observe the entire duration $[t_1, t_1 + \tau]$ for both $a$ and $b$, but only a part of the duration $[t_2, t_2 + \tau]$. An example is depicted in figure~\ref{fig::pair_censored_valid}. In this case, we would truncate the latter period for both individuals at the minimum of their two censoring times. In the example shown in figure~\ref{fig::pair_censored_valid}, only individual $b$ is censored at $t = 140$, so we would count events in the duration $[t_2, 140]$ for both $a$ and $b$ instead of the full duration. This works, because the time effects are the same in this case, so they still cancel out as shown in appendix~\ref{appendix::theta}.
\par\medskip
If the time of right-censoring for any of the two individuals in a pair occurs before $t_2$ instead, the pair remains invalid. In this case, we could not use any observation time for the required second duration. An example of this is shown in figure~\ref{fig::pair_censored_invalid}. Here, individual $a$ is censored at $t = 80$ and $t_2 = 120$. Individual $a$ can thus never serve as control for individual $b$ at $t_2$, making the pair invalid.
\par\medskip
Under certain circumstances, it is also possible to relax the requirement that $|t_1 - t_2| > \tau$. Consider the example shown in figure~\ref{fig::pair_overlap}. Here, the exposure times of individuals $a$ and $b$ are $100$ and $120$, respectively. Therefore, with $\tau = 40$, the risk periods of both individuals overlap during $[120, 140]$, making it impossible to use the full pair. However, we may use the non-overlapping regions to form the pair instead. Then, we would use the truncated periods of $[t_1, t_2]$ and $[t_1 + \tau, t_2]$ instead. This works as long as $t_1 \neq t_2$. Even using these truncated durations, the structural identity derived in appendix~\ref{appendix::theta} still holds, under the stated assumptions.
\par\medskip

\begin{figure}[!htb]
	\centering
	\includegraphics[width=1\linewidth]{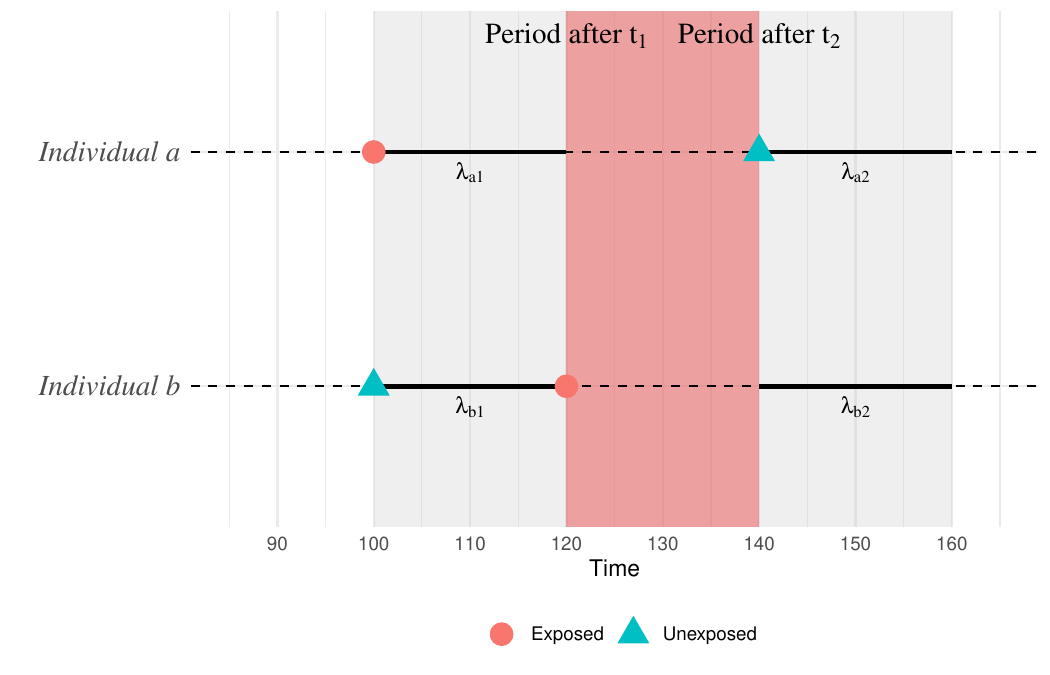}
	\caption{A graphical depiction of an exemplary valid pair $(a, b)$ with overlapping risk periods in the proposed symmetric pair matching design. Here, individual $a$ was exposed at $t_1 = 100$ and individual $b$ was exposed at $t_2 = 120$ with a risk period duration of $\tau = 40$. $\lambda_{a1}$, $\lambda_{a2}$, $\lambda_{b1}$ and $\lambda_{b2}$ denote the corresponding individual and time period specific event rates.}
	\label{fig::pair_overlap}
\end{figure}

Using these ``partial'' pairs may lead to a substantial gain in statistical efficiency when the observation period is short or if $\tau$ is large, relative to the observation period. In these cases, allowing partial pairs may substantially increase the number of available pairs. In usual applications with relatively short $\tau$, it will likely not have a substantial impact. The \texttt{SPMD} R package allows users to choose whether to use such pairs or not.

\FloatBarrier
\clearpage
\newpage

\section{Multiple Exposure Times per Individual} \label{appendix::exposure_times}

Throughout the main text and the consistency proof given in~\ref{appendix::consistency}, we assumed that each individual experiences at most one exposure instance. This assumption was made solely to simplify the notation and the dependency structure. In practice, the estimator naturally extends to multiple exposure periods by treating each exposure instance separately. Let:

\begin{equation}
	R_i = \bigcup_{k = 1}^{m_i} [T_{ik}, T_{ik} + \tau]
\end{equation}

denote the union of all exposure risk periods for individual $i$. A pair of exposure periods $(i, k)$ and $(j, l)$ is considered valid whenever:

\begin{equation}
	[T_{ik}, T_{ik} + \tau] \cap R_j = \emptyset,
\end{equation}

and

\begin{equation}
	[T_{jl}, T_{jl} + \tau] \cap R_i = \emptyset.
\end{equation}

These conditions ensure that both observation periods used as controls are entirely outside all exposure risk periods of the paired individual. Under the same multiplicative rate model used throughout the main text, the structural identity underlying the estimator continues to hold for every valid pair of exposure periods. Consider the example shown in figure~\ref{fig::pair_multiple}. Here, individual $a$ was exposed twice: once at $t = 10$ and once at $t = 70$. Individual $b$ was only exposed once at $t = 130$. Using just those two individuals, we could form the following two pairs of exposure risk periods:

\begin{equation}
	([10, 50], [130, 170]), \quad ([70, 110], [130, 170]),
\end{equation}

using the exposure risk period of individual $b$ twice. Had individual $b$ been exposed at, for example, $t = 20$, we would not be able to form the first pairing, but we would still be able to form the second one.
\par\medskip

\begin{figure}[!htb]
	\centering
	\includegraphics[width=1\linewidth]{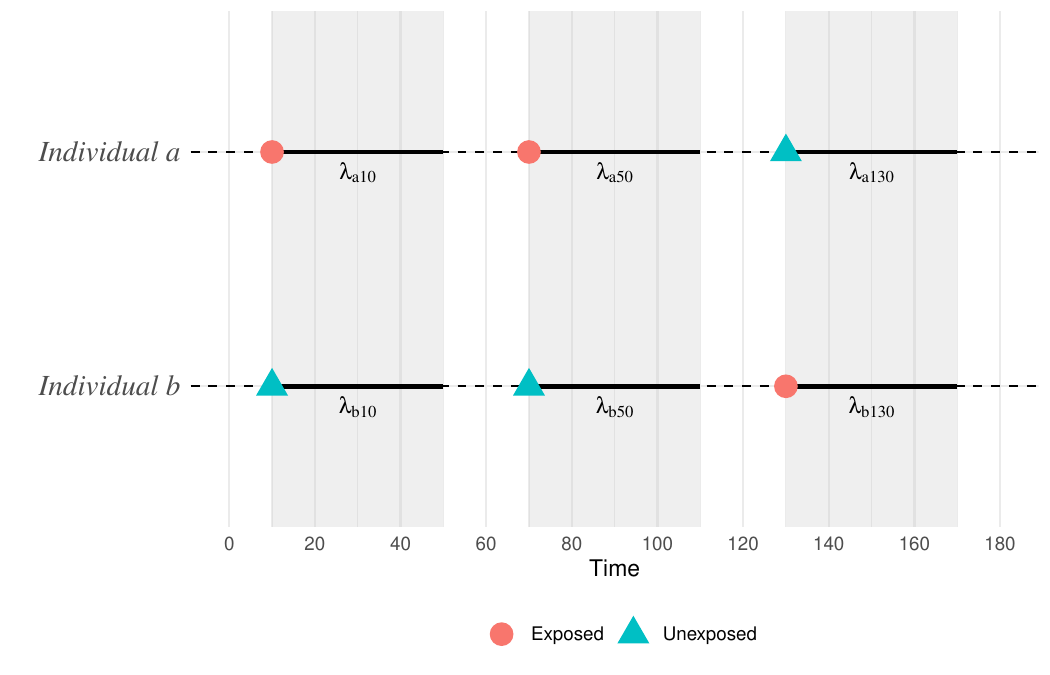}
	\caption{A graphical depiction of two individuals, $a$ and $b$. Here, individual $a$ was exposed at $t = 10$ and $t = 60$ and individual $b$ was exposed at $t = 130$ with a risk period duration of $\tau = 40$. The $\lambda$s denote the corresponding individual and time period specific event rates.}
	\label{fig::pair_multiple}
\end{figure}

Extending the consistency proof to this more general setting is primarily a technical problem, because repeated exposure episodes from the same individual introduce additional dependence between observations. We therefore present the proof only for the simpler single-exposure setting. The associated R package nevertheless supports multiple exposure periods per individual. This extension may also be applied using the relaxed validity criteria for pairs discussed in appendix~\ref{appendix::valid_pairs}. 

\FloatBarrier
\clearpage
\newpage

\section{Variance Estimation} \label{appendix::bootstrap}

In usual applications (e.g. when using all valid pairs or a large random subset thereof), individuals may contribute to multiple pairs simultaneously. Naive approximations of the variance ignoring these dependencies would underestimate the true variance. Instead, we recommend using person-level non-parametric bootstrapping. Here, $m$ samples of individuals of size $n$ are drawn with replacement from the original sample. The method is then applied to all $m$ samples. The standard deviation of the obtained point estimates may then be used to approximate the standard error of the point estimate, which may then be used to estimate 95\% confidence intervals using the normal approximation.
\par
When using all possible and valid pairs, this method may be computationally expensive, because the number of possible pairs grows by $\binom{n}{2}$. In some cases it may be feasible to obtain a point estimate, but infeasible to run thousands of naively generated bootstrap samples. Luckily, there is a computational trick that may be used to avoid creating all pairs on every bootstrap replication separately. Because a pair of two individuals $(a, b)$ is defined to be valid only if the periods defined by their respective exposure times $t_1$ and $t_2$ and $\tau$ are fully observed and do not overlap, a bootstrap sample cannot actually generate ``new'' pairs, that are not already included in the set of all pairs. It only re-weights existing pairs.
\par\medskip
We may therefore generate the set of all valid pairs once and then randomly draw weights for each pair from an appropriate distribution representing sampling without replacement. A weighted estimator may then be applied. Given $n$ inidividuals, the bootstrap weights are distributed as:

\begin{equation}
	(W_1, ..., W_n) \sim \text{Multinomial}\left(n; \frac{1}{n}, ..., \frac{1}{n}\right).
\end{equation}

In practice one may simply generate $n$ independent draws from a Poisson distribution with $\lambda = 1$ for $W_1, ..., W_n$ as an approximation. Every pair $(a, b)$ thus contributes weight $W_a W_b$, resulting in a weighted estimator of the form:

\begin{equation}
	\hat{\theta}_n^{boot} = \frac{1}{2} \ln\left(\frac{\sum_{(a, b) \in \mathcal{P}_n} W_a W_b X_{b1} X_{a2}}{\sum_{(a, b) \in \mathcal{P}_n} W_a W_b X_{a1} X_{b2}}\right),
\end{equation}

where $\mathcal{P}_n$ is the set of all valid pairs. By repeatedly drawing weights and applying the weighted estimator, multiple bootstrap estimates are created without the need for the creation of new pairs. The \texttt{SPMD} R package directly implements this.

\newpage

\section{Additional Simulation Results} \label{appendix::simulation}

The simulation study described in section~\ref{chap::simulation} of the main text was additionally repeated for different values of $\theta$ and with different sample sizes. The results are given below.

\begin{figure}[!htb]
	\centering
	\includegraphics[width=1\linewidth]{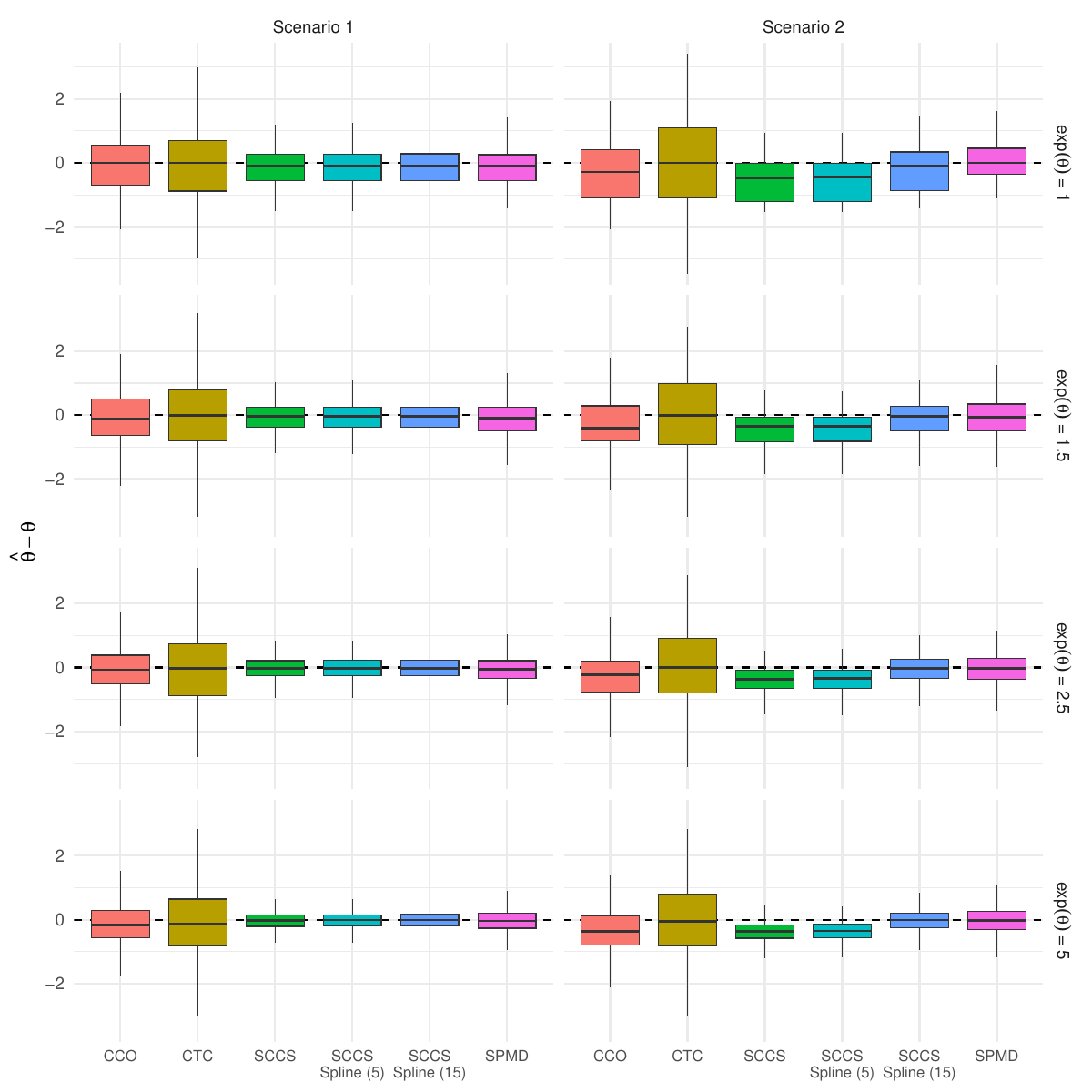}
	\caption{Boxplots of the deviations between the estimated and true $\theta$ for all considered methods, simulation scenarios and different true values of $\theta$. A sample size of $n = 10000$ was used with $1000$ simulation replications. Outliers are not shown.}
	\label{fig::sim_boxplots_n10000}
\end{figure}

\begin{figure}[!htb]
	\centering
	\includegraphics[width=1\linewidth]{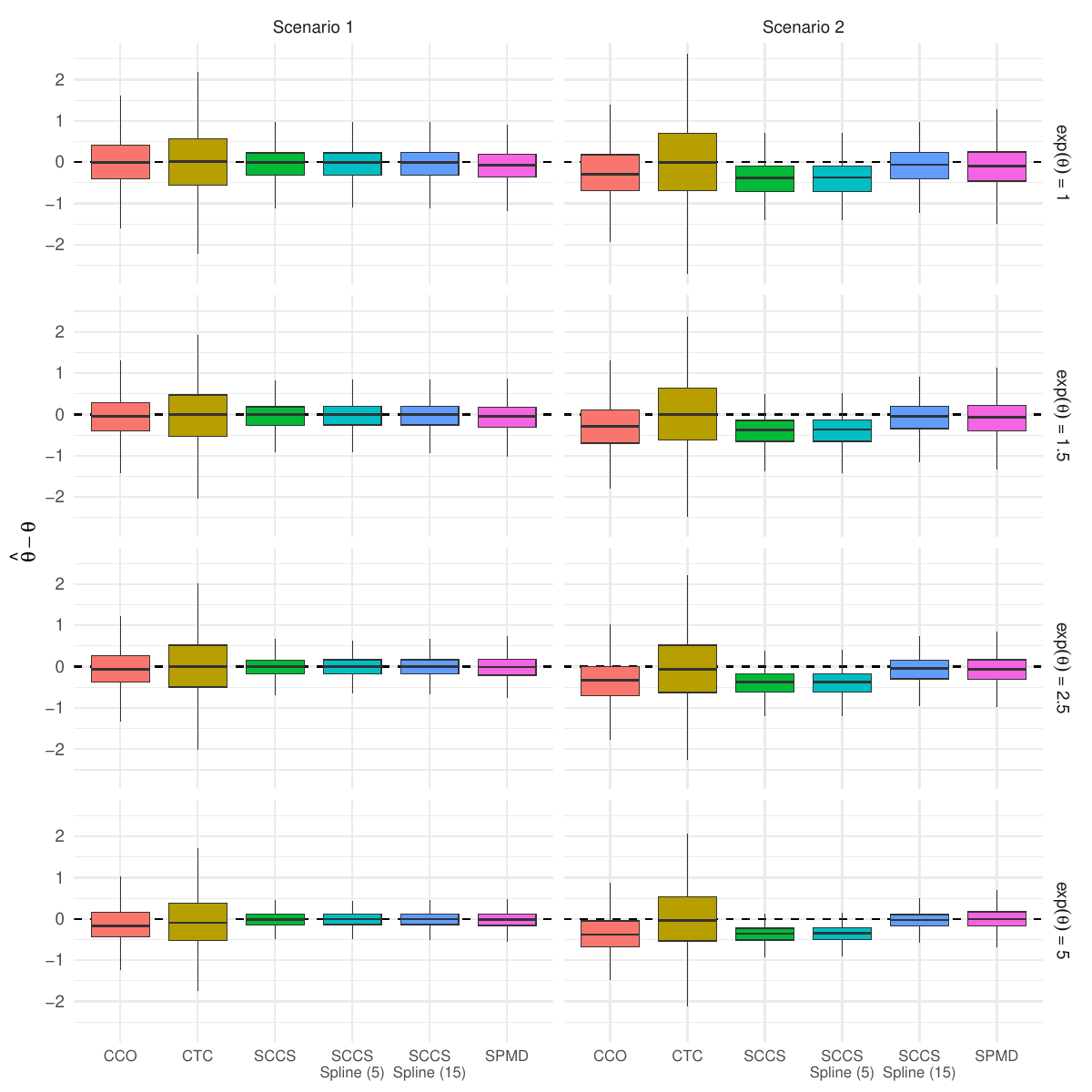}
	\caption{Boxplots of the deviations between the estimated and true $\theta$ for all considered methods, simulation scenarios and different true values of $\theta$. A sample size of $n = 20000$ was used with $1000$ simulation replications. Outliers are not shown.}
	\label{fig::sim_boxplots_n20000}
\end{figure}

\begin{figure}[!htb]
	\centering
	\includegraphics[width=1\linewidth]{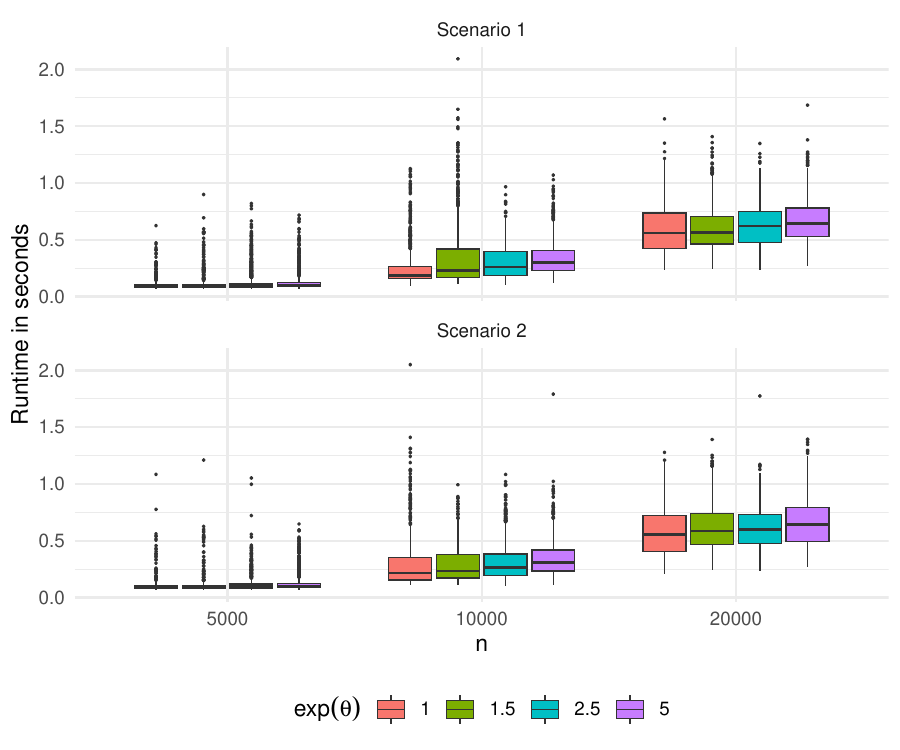}
	\caption{Boxplots showing the runtime of the symmetric pair matching implementation for all considered simulation scenarios and parameters.}
	\label{fig::spmd_runtime}
\end{figure}

\FloatBarrier

{\footnotesize
	\begin{longtable}{rllrrrrr}
		\toprule
		$\exp(\theta)$ & Scenario & Estimator & Bias & MCSE (Bias) & MSE & MCSE (MSE) & $n$ NA\\
		\midrule
		1.0 & Scenario 1 & Cox & -2.795 & 0.190 & 44.033 & 2.895 & 0\\
		1.0 & Scenario 1 & CCO & -0.602 & 0.372 & 138.291 & 6.099 & 38\\
		1.0 & Scenario 1 & CTC & -0.745 & 0.447 & 200.500 & 9.368 & 77\\
		1.0 & Scenario 1 & SCCS & -3.552 & 0.238 & 69.177 & 4.500 & 0\\
		1.0 & Scenario 1 & SCCS
		Spline (5) & -3.109 & 0.210 & 53.514 & 3.483 & 0\\
		1.0 & Scenario 1 & SCCS
		Spline (15) & -3.120 & 0.210 & 53.978 & 3.512 & 0\\
		1.0 & Scenario 1 & SPMD & 0.084 & 0.015 & 0.221 & 0.009 & 601\\
		1.0 & Scenario 2 & Cox & -4.428 & 0.225 & 70.009 & 3.395 & 0\\
		1.0 & Scenario 2 & CCO & -2.408 & 0.415 & 177.870 & 6.395 & 83\\
		1.0 & Scenario 2 & CTC & -0.975 & 0.523 & 274.364 & 11.985 & 146\\
		1.0 & Scenario 2 & SCCS & -5.838 & 0.275 & 109.400 & 5.262 & 0\\
		1.0 & Scenario 2 & SCCS
		Spline (5) & -5.151 & 0.242 & 85.058 & 4.096 & 0\\
		1.0 & Scenario 2 & SCCS
		Spline (15) & -4.920 & 0.246 & 84.421 & 4.056 & 0\\
		1.0 & Scenario 2 & SPMD & 0.095 & 0.013 & 0.172 & 0.008 & 830\\
		1.5 & Scenario 1 & Cox & -0.966 & 0.118 & 14.948 & 1.850 & 0\\
		1.5 & Scenario 1 & CCO & 1.659 & 0.283 & 82.574 & 5.114 & 18\\
		1.5 & Scenario 1 & CTC & 2.085 & 0.407 & 169.856 & 8.544 & 50\\
		1.5 & Scenario 1 & SCCS & -1.262 & 0.150 & 24.138 & 2.907 & 0\\
		1.5 & Scenario 1 & SCCS
		Spline (5) & -1.111 & 0.133 & 18.779 & 2.257 & 0\\
		1.5 & Scenario 1 & SCCS
		Spline (15) & -1.108 & 0.133 & 18.875 & 2.267 & 0\\
		1.5 & Scenario 1 & SPMD & -0.096 & 0.016 & 0.281 & 0.011 & 402\\
		1.5 & Scenario 2 & Cox & -2.431 & 0.181 & 38.664 & 2.833 & 0\\
		1.5 & Scenario 2 & CCO & 0.044 & 0.373 & 139.326 & 6.107 & 35\\
		1.5 & Scenario 2 & CTC & 2.242 & 0.487 & 242.058 & 11.506 & 95\\
		1.5 & Scenario 2 & SCCS & -3.392 & 0.222 & 60.903 & 4.385 & 0\\
		1.5 & Scenario 2 & SCCS
		Spline (5) & -3.032 & 0.197 & 47.999 & 3.454 & 0\\
		1.5 & Scenario 2 & SCCS
		Spline (15) & -2.735 & 0.199 & 47.174 & 3.403 & 0\\
		1.5 & Scenario 2 & SPMD & -0.182 & 0.015 & 0.256 & 0.009 & 720\\
		2.5 & Scenario 1 & Cox & -0.182 & 0.043 & 1.902 & 0.646 & 0\\
		2.5 & Scenario 1 & CCO & 2.419 & 0.240 & 63.161 & 4.448 & 3\\
		2.5 & Scenario 1 & CTC & 2.346 & 0.384 & 152.590 & 7.658 & 31\\
		2.5 & Scenario 1 & SCCS & -0.247 & 0.059 & 3.511 & 1.123 & 0\\
		2.5 & Scenario 1 & SCCS
		Spline (5) & -0.224 & 0.053 & 2.897 & 0.905 & 0\\
		2.5 & Scenario 1 & SCCS
		Spline (15) & -0.219 & 0.053 & 2.905 & 0.906 & 0\\
		2.5 & Scenario 1 & SPMD & -0.214 & 0.018 & 0.355 & 0.016 & 306\\
		2.5 & Scenario 2 & Cox & -0.974 & 0.116 & 14.327 & 1.869 & 0\\
		2.5 & Scenario 2 & CCO & 1.788 & 0.296 & 90.713 & 5.204 & 17\\
		2.5 & Scenario 2 & CTC & 3.791 & 0.449 & 215.488 & 10.128 & 69\\
		2.5 & Scenario 2 & SCCS & -1.537 & 0.142 & 22.450 & 2.871 & 0\\
		2.5 & Scenario 2 & SCCS
		Spline (5) & -1.404 & 0.126 & 17.818 & 2.267 & 0\\
		2.5 & Scenario 2 & SCCS
		Spline (15) & -1.088 & 0.127 & 17.351 & 2.225 & 0\\
		2.5 & Scenario 2 & SPMD & -0.433 & 0.017 & 0.473 & 0.018 & 604\\
		5.0 & Scenario 1 & Cox & -0.062 & 0.013 & 0.163 & 0.011 & 0\\
		5.0 & Scenario 1 & CCO & 2.408 & 0.210 & 50.012 & 3.836 & 0\\
		5.0 & Scenario 1 & CTC & 2.535 & 0.352 & 129.898 & 6.967 & 25\\
		5.0 & Scenario 1 & SCCS & -0.069 & 0.013 & 0.173 & 0.012 & 0\\
		5.0 & Scenario 1 & SCCS
		Spline (5) & -0.061 & 0.013 & 0.176 & 0.012 & 0\\
		5.0 & Scenario 1 & SCCS
		Spline (15) & -0.051 & 0.013 & 0.179 & 0.013 & 0\\
		5.0 & Scenario 1 & SPMD & -0.186 & 0.015 & 0.270 & 0.016 & 232\\
		5.0 & Scenario 2 & Cox & -0.132 & 0.023 & 0.530 & 0.289 & 0\\
		5.0 & Scenario 2 & CCO & 2.723 & 0.234 & 62.229 & 4.188 & 1\\
		5.0 & Scenario 2 & CTC & 4.027 & 0.407 & 181.887 & 9.330 & 60\\
		5.0 & Scenario 2 & SCCS & -0.482 & 0.025 & 0.882 & 0.427 & 0\\
		5.0 & Scenario 2 & SCCS
		Spline (5) & -0.470 & 0.023 & 0.760 & 0.311 & 0\\
		5.0 & Scenario 2 & SCCS
		Spline (15) & -0.140 & 0.024 & 0.604 & 0.322 & 0\\
		5.0 & Scenario 2 & SPMD & -0.453 & 0.016 & 0.475 & 0.022 & 468\\
		\bottomrule
		\caption{Average bias and mean-squared error (MSE), as well as their associated Monte-Carlo standard errors and the number of infinite or missing estimates ($n$ NA) for every simulation scenario, estimator and different values of $\theta$. A sample size of $n = 5000$ with $1000$ simulation replications was used throughout. Infinite or missing estimates were removed before further calculations.}
	\end{longtable}
}

\newpage

{\footnotesize
	\begin{longtable}{rllrrrrr}
		\toprule
		$\exp(\theta)$ & Scenario & Estimator & Bias & MCSE (Bias) & MSE & MCSE (MSE) & $n$ NA\\
		\midrule
		1.0 & Scenario 1 & Cox & -0.430 & 0.068 & 4.809 & 1.016 & 0\\
		1.0 & Scenario 1 & CCO & -0.042 & 0.166 & 27.592 & 3.203 & 1\\
		1.0 & Scenario 1 & CTC & -0.051 & 0.208 & 43.177 & 4.088 & 2\\
		1.0 & Scenario 1 & SCCS & -0.522 & 0.086 & 7.602 & 1.606 & 0\\
		1.0 & Scenario 1 & SCCS
		Spline (5) & -0.472 & 0.076 & 5.951 & 1.242 & 0\\
		1.0 & Scenario 1 & SCCS
		Spline (15) & -0.473 & 0.076 & 6.023 & 1.257 & 0\\
		1.0 & Scenario 1 & SPMD & -0.125 & 0.018 & 0.323 & 0.012 & 144\\
		1.0 & Scenario 2 & Cox & -1.447 & 0.140 & 21.712 & 2.138 & 0\\
		1.0 & Scenario 2 & CCO & -1.396 & 0.251 & 65.017 & 4.702 & 6\\
		1.0 & Scenario 2 & CTC & -0.451 & 0.298 & 88.757 & 5.979 & 12\\
		1.0 & Scenario 2 & SCCS & -2.173 & 0.174 & 34.923 & 3.361 & 0\\
		1.0 & Scenario 2 & SCCS
		Spline (5) & -1.955 & 0.153 & 27.355 & 2.626 & 0\\
		1.0 & Scenario 2 & SCCS
		Spline (15) & -1.649 & 0.155 & 26.821 & 2.585 & 0\\
		1.0 & Scenario 2 & SPMD & 0.053 & 0.017 & 0.289 & 0.011 & 340\\
		1.5 & Scenario 1 & Cox & -0.159 & 0.035 & 1.246 & 0.495 & 0\\
		1.5 & Scenario 1 & CCO & 0.282 & 0.106 & 11.307 & 2.033 & 0\\
		1.5 & Scenario 1 & CTC & 0.314 & 0.168 & 28.417 & 3.119 & 2\\
		1.5 & Scenario 1 & SCCS & -0.180 & 0.042 & 1.775 & 0.755 & 0\\
		1.5 & Scenario 1 & SCCS
		Spline (5) & -0.168 & 0.038 & 1.478 & 0.608 & 0\\
		1.5 & Scenario 1 & SCCS
		Spline (15) & -0.165 & 0.038 & 1.478 & 0.607 & 0\\
		1.5 & Scenario 1 & SPMD & -0.131 & 0.018 & 0.330 & 0.015 & 36\\
		1.5 & Scenario 2 & Cox & -0.456 & 0.078 & 6.237 & 1.199 & 0\\
		1.5 & Scenario 2 & CCO & -0.157 & 0.181 & 32.771 & 3.471 & 1\\
		1.5 & Scenario 2 & CTC & 0.929 & 0.263 & 69.701 & 5.007 & 5\\
		1.5 & Scenario 2 & SCCS & -0.910 & 0.095 & 9.930 & 1.867 & 0\\
		1.5 & Scenario 2 & SCCS
		Spline (5) & -0.849 & 0.085 & 7.962 & 1.480 & 0\\
		1.5 & Scenario 2 & SCCS
		Spline (15) & -0.525 & 0.086 & 7.654 & 1.450 & 0\\
		1.5 & Scenario 2 & SPMD & -0.088 & 0.018 & 0.346 & 0.014 & 158\\
		2.5 & Scenario 1 & Cox & -0.054 & 0.012 & 0.141 & 0.009 & 0\\
		2.5 & Scenario 1 & CCO & 0.433 & 0.100 & 10.135 & 1.871 & 0\\
		2.5 & Scenario 1 & CTC & 0.552 & 0.173 & 30.294 & 3.182 & 1\\
		2.5 & Scenario 1 & SCCS & -0.059 & 0.012 & 0.150 & 0.010 & 0\\
		2.5 & Scenario 1 & SCCS
		Spline (5) & -0.056 & 0.012 & 0.151 & 0.010 & 0\\
		2.5 & Scenario 1 & SCCS
		Spline (15) & -0.052 & 0.012 & 0.153 & 0.010 & 0\\
		2.5 & Scenario 1 & SPMD & -0.073 & 0.015 & 0.221 & 0.012 & 8\\
		2.5 & Scenario 2 & Cox & -0.066 & 0.014 & 0.212 & 0.014 & 0\\
		2.5 & Scenario 2 & CCO & 0.496 & 0.126 & 16.216 & 2.357 & 0\\
		2.5 & Scenario 2 & CTC & 1.414 & 0.244 & 61.491 & 4.823 & 4\\
		2.5 & Scenario 2 & SCCS & -0.439 & 0.025 & 0.793 & 0.399 & 0\\
		2.5 & Scenario 2 & SCCS
		Spline (5) & -0.429 & 0.023 & 0.711 & 0.322 & 0\\
		2.5 & Scenario 2 & SCCS
		Spline (15) & -0.104 & 0.023 & 0.544 & 0.310 & 0\\
		2.5 & Scenario 2 & SPMD & -0.092 & 0.017 & 0.308 & 0.015 & 73\\
		5.0 & Scenario 1 & Cox & -0.035 & 0.008 & 0.072 & 0.004 & 0\\
		5.0 & Scenario 1 & CCO & 0.393 & 0.100 & 10.063 & 1.802 & 0\\
		5.0 & Scenario 1 & CTC & 0.216 & 0.153 & 23.464 & 2.903 & 2\\
		5.0 & Scenario 1 & SCCS & -0.047 & 0.009 & 0.081 & 0.006 & 0\\
		5.0 & Scenario 1 & SCCS
		Spline (5) & -0.038 & 0.009 & 0.081 & 0.005 & 0\\
		5.0 & Scenario 1 & SCCS
		Spline (15) & -0.033 & 0.009 & 0.082 & 0.005 & 0\\
		5.0 & Scenario 1 & SPMD & -0.028 & 0.012 & 0.133 & 0.008 & 3\\
		5.0 & Scenario 2 & Cox & -0.038 & 0.010 & 0.105 & 0.006 & 0\\
		5.0 & Scenario 2 & CCO & 0.234 & 0.104 & 10.876 & 1.863 & 0\\
		5.0 & Scenario 2 & CTC & 1.095 & 0.193 & 38.353 & 3.813 & 4\\
		5.0 & Scenario 2 & SCCS & -0.389 & 0.010 & 0.259 & 0.012 & 0\\
		5.0 & Scenario 2 & SCCS
		Spline (5) & -0.380 & 0.010 & 0.253 & 0.012 & 0\\
		5.0 & Scenario 2 & SCCS
		Spline (15) & -0.049 & 0.011 & 0.127 & 0.007 & 0\\
		5.0 & Scenario 2 & SPMD & -0.034 & 0.015 & 0.213 & 0.011 & 73\\
		\bottomrule
		\caption{Average bias and mean-squared error (MSE), as well as their associated Monte-Carlo standard errors and the number of infinite or missing estimates ($n$ NA) for every simulation scenario, estimator and different values of $\theta$. A sample size of $n = 10000$ with $1000$ simulation replications was used throughout. Infinite or missing estimates were removed before further calculations.}
	\end{longtable}
}

\newpage

{\footnotesize
	\begin{longtable}{rllrrrrr}
		\toprule
		$\exp(\theta)$ & Scenario & Estimator & Bias & MCSE (Bias) & MSE & MCSE (MSE) & $n$ NA\\
		\midrule
		1.0 & Scenario 1 & Cox & -0.064 & 0.014 & 0.191 & 0.013 & 0\\
		1.0 & Scenario 1 & CCO & -0.004 & 0.028 & 0.797 & 0.408 & 0\\
		1.0 & Scenario 1 & CTC & 0.021 & 0.034 & 1.122 & 0.341 & 0\\
		1.0 & Scenario 1 & SCCS & -0.075 & 0.014 & 0.203 & 0.014 & 0\\
		1.0 & Scenario 1 & SCCS
		Spline (5) & -0.072 & 0.014 & 0.202 & 0.014 & 0\\
		1.0 & Scenario 1 & SCCS
		Spline (15) & -0.072 & 0.014 & 0.203 & 0.014 & 0\\
		1.0 & Scenario 1 & SPMD & -0.128 & 0.015 & 0.235 & 0.014 & 7\\
		1.0 & Scenario 2 & Cox & -0.189 & 0.040 & 1.670 & 0.573 & 0\\
		1.0 & Scenario 2 & CCO & -0.449 & 0.078 & 6.305 & 1.515 & 0\\
		1.0 & Scenario 2 & CTC & -0.042 & 0.095 & 9.015 & 1.552 & 0\\
		1.0 & Scenario 2 & SCCS & -0.562 & 0.048 & 2.640 & 0.886 & 0\\
		1.0 & Scenario 2 & SCCS
		Spline (5) & -0.541 & 0.043 & 2.134 & 0.683 & 0\\
		1.0 & Scenario 2 & SCCS
		Spline (15) & -0.221 & 0.044 & 1.959 & 0.681 & 0\\
		1.0 & Scenario 2 & SPMD & -0.127 & 0.017 & 0.304 & 0.014 & 42\\
		1.5 & Scenario 1 & Cox & -0.044 & 0.011 & 0.117 & 0.008 & 0\\
		1.5 & Scenario 1 & CCO & 0.003 & 0.027 & 0.705 & 0.392 & 0\\
		1.5 & Scenario 1 & CTC & 0.007 & 0.035 & 1.245 & 0.442 & 0\\
		1.5 & Scenario 1 & SCCS & -0.069 & 0.022 & 0.502 & 0.378 & 0\\
		1.5 & Scenario 1 & SCCS
		Spline (5) & -0.062 & 0.020 & 0.393 & 0.270 & 0\\
		1.5 & Scenario 1 & SCCS
		Spline (15) & -0.062 & 0.020 & 0.394 & 0.270 & 0\\
		1.5 & Scenario 1 & SPMD & -0.077 & 0.012 & 0.147 & 0.008 & 1\\
		1.5 & Scenario 2 & Cox & -0.075 & 0.014 & 0.189 & 0.013 & 0\\
		1.5 & Scenario 2 & CCO & -0.239 & 0.040 & 1.652 & 0.678 & 0\\
		1.5 & Scenario 2 & CTC & 0.140 & 0.064 & 4.067 & 1.017 & 0\\
		1.5 & Scenario 2 & SCCS & -0.434 & 0.014 & 0.376 & 0.021 & 0\\
		1.5 & Scenario 2 & SCCS
		Spline (5) & -0.428 & 0.014 & 0.371 & 0.021 & 0\\
		1.5 & Scenario 2 & SCCS
		Spline (15) & -0.105 & 0.014 & 0.204 & 0.014 & 0\\
		1.5 & Scenario 2 & SPMD & -0.123 & 0.015 & 0.251 & 0.014 & 9\\
		2.5 & Scenario 1 & Cox & -0.014 & 0.008 & 0.060 & 0.003 & 0\\
		2.5 & Scenario 1 & CCO & -0.035 & 0.016 & 0.242 & 0.012 & 0\\
		2.5 & Scenario 1 & CTC & 0.039 & 0.035 & 1.238 & 0.435 & 0\\
		2.5 & Scenario 1 & SCCS & -0.018 & 0.008 & 0.064 & 0.003 & 0\\
		2.5 & Scenario 1 & SCCS
		Spline (5) & -0.013 & 0.008 & 0.065 & 0.003 & 0\\
		2.5 & Scenario 1 & SCCS
		Spline (15) & -0.011 & 0.008 & 0.065 & 0.003 & 0\\
		2.5 & Scenario 1 & SPMD & -0.019 & 0.009 & 0.081 & 0.004 & 0\\
		2.5 & Scenario 2 & Cox & -0.059 & 0.010 & 0.102 & 0.006 & 0\\
		2.5 & Scenario 2 & CCO & -0.283 & 0.038 & 1.543 & 0.643 & 0\\
		2.5 & Scenario 2 & CTC & 0.125 & 0.079 & 6.188 & 1.297 & 0\\
		2.5 & Scenario 2 & SCCS & -0.412 & 0.010 & 0.274 & 0.011 & 0\\
		2.5 & Scenario 2 & SCCS
		Spline (5) & -0.405 & 0.010 & 0.269 & 0.011 & 0\\
		2.5 & Scenario 2 & SCCS
		Spline (15) & -0.083 & 0.011 & 0.117 & 0.006 & 0\\
		2.5 & Scenario 2 & SPMD & -0.082 & 0.012 & 0.150 & 0.009 & 0\\
		5.0 & Scenario 1 & Cox & -0.018 & 0.006 & 0.032 & 0.002 & 0\\
		5.0 & Scenario 1 & CCO & -0.119 & 0.015 & 0.238 & 0.013 & 0\\
		5.0 & Scenario 1 & CTC & -0.077 & 0.023 & 0.524 & 0.028 & 0\\
		5.0 & Scenario 1 & SCCS & -0.022 & 0.006 & 0.035 & 0.002 & 0\\
		5.0 & Scenario 1 & SCCS
		Spline (5) & -0.017 & 0.006 & 0.035 & 0.002 & 0\\
		5.0 & Scenario 1 & SCCS
		Spline (15) & -0.014 & 0.006 & 0.035 & 0.002 & 0\\
		5.0 & Scenario 1 & SPMD & -0.018 & 0.007 & 0.049 & 0.002 & 0\\
		5.0 & Scenario 2 & Cox & -0.016 & 0.006 & 0.042 & 0.002 & 0\\
		5.0 & Scenario 2 & CCO & -0.338 & 0.016 & 0.366 & 0.015 & 0\\
		5.0 & Scenario 2 & CTC & 0.198 & 0.067 & 4.463 & 1.094 & 0\\
		5.0 & Scenario 2 & SCCS & -0.372 & 0.007 & 0.183 & 0.006 & 0\\
		5.0 & Scenario 2 & SCCS
		Spline (5) & -0.364 & 0.007 & 0.178 & 0.006 & 0\\
		5.0 & Scenario 2 & SCCS
		Spline (15) & -0.036 & 0.007 & 0.051 & 0.003 & 0\\
		5.0 & Scenario 2 & SPMD & 0.004 & 0.009 & 0.079 & 0.004 & 1\\
		\bottomrule
		\caption{Average bias and mean-squared error (MSE), as well as their associated Monte-Carlo standard errors and the number of infinite or missing estimates ($n$ NA) for every simulation scenario, estimator and different values of $\theta$. A sample size of $n = 5000$ with $2000$ simulation replications was used throughout. Infinite or missing estimates were removed before further calculations.}
	\end{longtable}
}

\end{document}